\documentclass[11pt]{article}

\usepackage{arxiv}

\usepackage[utf8]{inputenc} % allow utf-8 input
\usepackage[T1]{fontenc}    % use 8-bit T1 fonts
\usepackage{hyperref}       % hyperlinks
\usepackage{url}            % simple URL typesetting
\usepackage{microtype}      % microtypography
\usepackage{cite}
\usepackage{natbib}
\usepackage{doi}

\usepackage{booktabs}
\usepackage{amsmath,amssymb,amsfonts}
\usepackage{cleveref}       % smart cross-referencing
\usepackage{algorithmic}
\usepackage{graphicx}
\usepackage{textcomp}
\usepackage{xcolor}
\usepackage{comment}
\usepackage{amsmath,amsthm,amsfonts,amssymb,amscd}

\newtheorem{theorem}{Theorem}
\newtheorem{example}{Example}[section]

\newtheorem{lemma}[example]{Lemma}

\usepackage{blindtext}
\usepackage{abraces}
\usepackage{float}
\usepackage{resizegather}
\usepackage{nicefrac}
\usepackage{graphicx,subcaption}
\usepackage{multirow}

\definecolor{color4}{RGB}{179, 43, 59}
\definecolor{green2}{rgb}{0, 153, 153}
\definecolor{wpurple}{rgb}{.341, .035, .953}

\author{ \href{https://orcid.org/0000-0001-6679-2117}{\includegraphics[scale=0.06]{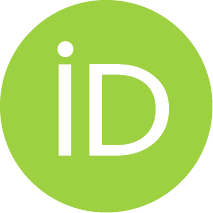}\hspace{1mm}Thiago S. Gomides}%\thanks{Use footnote for providing further
		\\
	School of Computer Science\\
	Carleton University\\
	Ottawa, ON, Canada \\
	\texttt{thiagodasilvagomides@cmail.carleton.ca} \\
	\And
	\href{https://orcid.org/0000-0002-8959-4428}{\includegraphics[scale=0.06]{orcid.pdf}\hspace{1mm}Evangelos Kranakis} \\
	School of Computer Science\\
	Carleton University\\
	Ottawa, ON, Canada \\
	\texttt{kranakis@scs.carleton.ca} \\
		\And
	\href{https://orcid.org/0000-0003-4686-9433}{\includegraphics[scale=0.06]{orcid.pdf}\hspace{1mm}Ioannis Lambadaris} \\
	Department of Systems and Computer Engineering\\
	Carleton University\\
	Ottawa, ON, Canada \\
	\texttt{ioannis@sce.carleton.ca} \\
		\And
	~~~~~~~~~~~Yannis Viniotis \\
	~~~~~~~~~~~Department of Electrical and Computer Engineering\\
	~~~~~~~~~~~North Carolina State University\\
	~~~~~~~~~~~Raleigh, NC, USA \\
	~~~~~~~~~~~\texttt{candice@ncsu.edu} \\
}

\renewcommand{\shorttitle}{Optimal Control for Platooning in Vehicular Networks}

\hypersetup{
pdftitle={Optimal Control for Platooning in Vehicular Networks},
pdfsubject={},
pdfauthor={Thiago S. Gomides, Evangelos Kranakis, Ioannis Lambadaris, Yannis Viniotis},
pdfkeywords={Truck platooning, Optimal control of queues.},
}

\newif\ifbulletlist
\newif\iftext
\texttrue % comment to hide text

\title{Optimal Control for Platooning in Vehicular Networks}
\date{}

\begin{document}
\maketitle

\begin{abstract}
As the automotive industry is developing autonomous driving systems and vehicular networks, attention to truck platooning has increased as a way to reduce costs (fuel consumption) and improve efficiency in the highway. Recent research in this area has focused mainly on the aerodynamics, network stability, and longitudinal control of platoons. However, the system aspects (e.g., platoon coordination) are still not well explored. In this paper, we formulate a platooning coordination problem and study whether  trucks waiting at an initial location (station) should wait for a platoon to arrive in order to leave. Arrivals of trucks at the station and platoons by the station are modelled by independent Bernoulli distributions. Next we use the theory of Markov Decision Processes to formulate the dispatching control problem and derive the optimal policy governing the dispatching of trucks with platoons. We show that the policy that minimizes an average cost function at the station is of threshold type. Numerical results for the average cost case are presented. They are consistent with the optimal ones.


\end{abstract}

\section{Introduction} 
\ifbulletlist
{\color{red} 
\begin{enumerate}
    \item Truck platooning definition and its potential.
    \item Importance of studying platooning under concrete math models.
\end{enumerate}
}
\fi %bulletlist

Truck platooning is the practice of virtually connecting two or more automated trucks forming convoys, where trucks follow one another closely. This practice has recently gained attention as the automotive industry develops toward autonomous driving systems and vehicular networks \cite{Adler2020}. 
%By definition,  
It holds great potential to make traffic more efficient and clean. In particular, allowing close-distance driving mitigates the effects of aerodynamic drag, which in turn leads to a substantial reduction in fuel consumption \cite{Zabat1995TheAP}. Platooning also optimizes highway use, reduces travel times and enhances transportation safety.
 
%It is worth mentioning that 
The benefits of platooning may vary due to the complex and dynamic behaviour of trucks and the % \ek{resulting} 
resulting traffic. For instance, this potential depends on several aspects, such as the inter-vehicle gap in a platoon, the travel speed, the cost of platooning formation, and others \cite{Zabat1995TheAP}. Therefore, it is crucial to study and understand platooning under concrete mathematical models.

\subsection{Related Work}
 \ifbulletlist
{\color{red} 
\begin{enumerate}
    \item Overview of recent platooning research. Coordination techniques haven't been well explored.
    \item Description of platooning coordination problems listing some relevant related work.  
    \item Example of three related works that approached a similar but not the same problem.
\end{enumerate}
}
\fi %bulletlist
Most of the research efforts so far have been concerned with studying the aerodynamic aspects of platooning \cite{Zabat1995TheAP,8604977}, the cooperative longitudinal control of trucks \cite{7497531,8957499,6588305}, and the stability of the platoons from a network perspective \cite{9416853,8967210}. However, the system aspects %of platooning 
(e.g., platoon coordination) are still not well explored \cite{Adler2020}, especially concerning optimal control.

Previous works on platooning coordination considered the dispatching control of trucks waiting in a station/hub \cite{ZHANG20171,Adler2020,Johansson2020TruckPF}. In these works, trucks arrive at the station following a random process (e.g., Poisson or Bernoulli), and the station decides whether they should wait to form  platoons. If the station holds them, it may build platoons with many trucks, which reduces fuel consumption. However, forcing many trucks to wait at the station incurs high transportation delay cost. These works investigated the optimal dispatching control at the station that minimizes an average cost function.

In \cite{ZHANG20171}, the authors studied optimal platoon coordination at a highway junction (hub). Their model consisted of two trucks  arriving at the hub with stochastic arrival times. If they arrive at the same time, they form a platoon. One truck may have to wait for the other when their arrival times differ, which incurs a waiting cost. The authors proved that it is optimal to build a platoon only when the arrival time of each truck differs by less than a threshold. 

In \cite{Adler2020}, the authors studied platooning coordination of multiple trucks at a station, where the arrivals are Poisson distributed. The station decides whether to hold trucks %\ek{in order}
in order to build platoons. The authors compared different truck dispatching policies 
%\ek{(to govern) governing}
governing the station under energy-delay tradeoff considerations. They proved the optimality of threshold policies to control station dispatching, where the station %\ek{(dispatch) dispatches} 
dispatches all trucks whenever the number of waiting trucks in the station grows above the threshold. In \cite{Johansson2020TruckPF}, the authors proposed a similar model to the one in \cite{Adler2020} under the assumption that arrivals of trucks are i.i.d and their distribution is known by the station. %They proved the optimality of a one time-step look-ahead policy. 
They proved that the optimal policy is of the type "one time-step look-ahead".

\subsection{Model Novelty and Main Contributions}
 \ifbulletlist
{\color{red} 
\begin{enumerate}
    \item Description of the problem approached in the paper. 
    \item Description of the techniques we use and our main contributions.
    \item Description of the following sections.
    \item Proofs are omitted and can be found in \cite{arxiv}.
\end{enumerate}
}
\fi %bulletlist

%This paper studies platooning coordination from a novel perspective.% compared to the ones in \cite{ZHANG20171, Adler2020, Johansson2020TruckPF}. 
In this paper, we study the dispatching and formation of platoons from a novel perspective. In particular, like \cite{Adler2020}, \cite{Johansson2020TruckPF} (and unlike \cite{ZHANG20171}), in our model we consider a waiting station that can hold multiple trucks. However, unlike \cite{Adler2020} and \cite{Johansson2020TruckPF},
we assume that trucks arrive at the station, while platoons arrive alongside the station.  Therefore, we dispatch trucks with arriving platoons, as opposed to forming platoons among waiting trucks.   Our cost function is also different.  %We consider a discrete time mathematical model. 

We formalize   the dispatching actions at the station as an optimal control problem. We then use dynamic programming to analyze it.   The main contributions of this work are as follows:
\begin{itemize}
	%\item We formulate and formalize the dispatching control problem of vehicles with platoons.
	\item We derive the optimal dispatching policy % \ek{(to govern) governing} 
 governing the system under finite and infinite horizon discounted cost criteria. We show that the policies are of threshold type with a finite threshold. 
	\item We use Lippman's \cite{10.2307/2629409} results to  derive the optimal policy for the average cost criterion.
	\item We present numerical results for the average cost case.
\end{itemize}

 The paper is organized as follows. In Section \ref{sec:formulation}, we formulate the  dispatching trucks to arriving platoons problem.  We also present the Markov Decision Problem. We characterize the optimal control policy for the discounted cost (with finite and infinite horizons) in Section \ref{sec:optimal}. In Section \ref{sec:longrun}, we derive the optimal policy for the average cost problem. We present some numerical results for the average cost problem in Section \ref{sec:results}. Our conclusion and suggestions for future work are presented in Section \ref{sec:conclusion}. %Due to space constraints, all but one  proofs have been omitted and can be found in \cite{arxiv}.

\section{Model and Control Problem Formulation}\label{sec:formulation}
%\yv{This should be a section, not a subsection. Since it is short, it can be combined with the next section, as Model and Control Problem Formulation }
\ifbulletlist
{\color{red} 
\begin{enumerate}
    \item Definition of the station, trucks/platoons arrivals, and dispatching control.
    \item Description of holding and dispatching actions and their benefits/costs. %Platooning is preferable whenever possible.
    \item Assumption of platoons providing the same energy reduction. Introduction of $\kappa>$ holding cost. 
    \item A practical argument that justifies the assumption.
\end{enumerate}
}
\fi %bulletlist

We consider the platooning system shown in Figure \ref{fig:overview}.  %\tgq{we dispatch at most one at a time}
\begin{figure}[h]
\centering
\includegraphics[width=.8\linewidth]{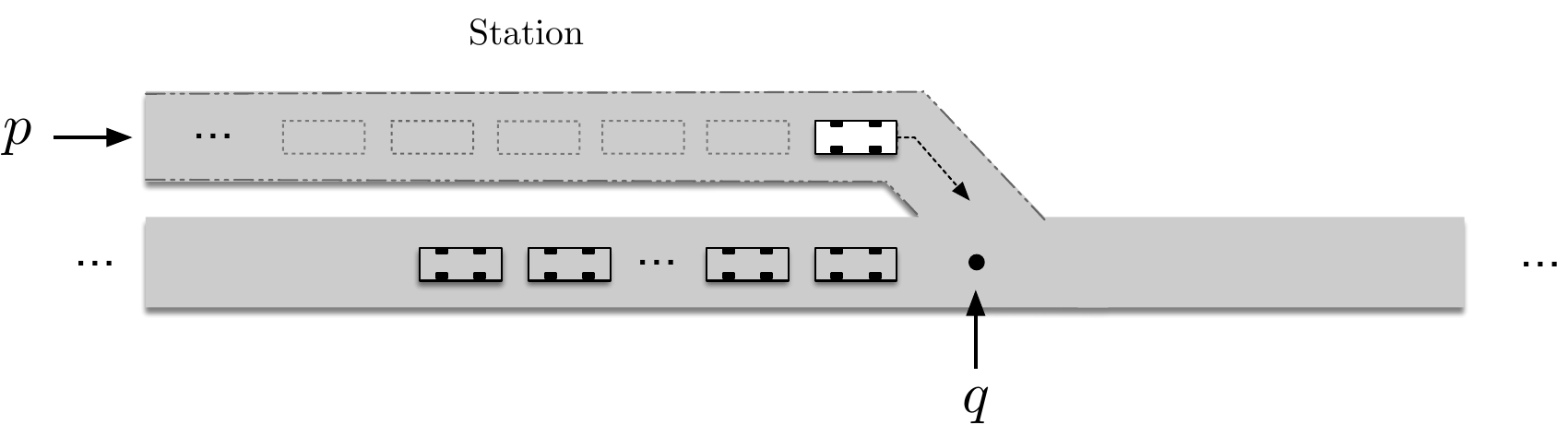}
\caption{System model.}
\label{fig:overview}
\end{figure}

Define the station as a location where trucks arrive and can wait to join platoons. Platoons arrive alongside the station. The station decides how trucks are dispatched from the station. We assume that only one truck at a time can be dispatched.

The system operates in discrete time.  Truck arrivals are Bernoulli with parameter $p$. Trucks join the same queue upon arrival. 
Platoon arrivals are Bernoulli with  parameter $q$.
In order to avoid trivialities, we assume that $0<p<1$ and $0<q<1$.

We assume truck and platoon arrivals are a sequence of ordered events in one time slot. A truck arrival (if any) always takes place earlier than a platoon arrival (if any). This assumption guarantees that if the station is empty and a truck arrives, it can join a platoon if any arrives.

\subsection{The Markovian Decision problem}

Let $x_n$ denote the number of trucks waiting at the station at slot $n$, $n \in \{1,2,\dots\}$. $x_n$  is the state of the system and $X = \{0,1,2,\dots\}$ is the state space of the system. 

Define the events as the combinations of arrivals (of a truck or a platoon) that may occur during a time slot. Each event has a given probability and a state operator (mapping $X$ into $X$). Transitions among the states are described in Table~\ref{tab:dp}.
\begin{table}[H]
\centering
\begin{tabular}{llr}       
\toprule
\textbf{Events}  & \textbf{State Operators}  & \textbf{Probabilities}   \\ 
\hline
No arrivals.     & $Z(x) = x.$  & $(1-p)(1-q)$                         \\
A platoon arrives and no truck arrives.  & $P(x) = x.$  & $(1-p)q$                           \\
A truck arrives and no platoon arrives. & $V(x) = x+1.$               & $p(1-q)$             \\
A platoon and a truck arrive.    & $B(x) = x+1.$                & $pq$                 \\
\bottomrule
\end{tabular}
\caption{Events, State Operators, and Probabilities.}
\label{tab:dp}
\end{table}
Define the action space $\mathcal{A}=\{H,D\}$, where the action operator $H$ represents the action of holding a truck at the station, and $D$ denotes the dispatching action. We have:
\begin{align*}
H (x) &=  x,  \operatorname{\textit{     ~~~~~dom }}  H = X. \\
D (x) &=  x - 1, \operatorname{\textit{ dom }}  D = \{x \in X: x \geq 1\}.
\end{align*}
$\mathcal{A}(x_n)$ represents the available actions for a given state $x_n$ at time slot $n$.
%The action space is $\mathcal{A} = \{H,D\}$.
%where $x$ is the system state after an event $E_i$ takes place. $H$ denotes the action of {holding a vehicle at the station} during the current time slot, and $D$ is the action of {dispatching a vehicle from the station}. 
%The action space is $\mathcal{A} = \{H,D\}$. %, and $\mathcal{A} = \{H,D\}$.

Dispatching trucks with platoons reduces energy (fuel) consumption. Holding trucks at the station (when waiting for platoons) %\ek{(incur) incurs} 
incurs transportation delay costs. Trucks dispatched without platoons pay the entire transportation cost. 

We assume all arriving platoons provide the same energy reduction. It simplifies the model as the cost of dispatching a truck by platooning (or not) becomes deterministic. We introduce a real-valued constant $\kappa$ representing the cost of dispatching a truck without a platoon.

We formalize these assumptions mathematically with $c(x_n,a_n)$, the instantaneous cost as a function of the system state $x_n$ when taking the action $a_n$ at time slot $n$. More specifically,
satisfies:
\begin{align}
c(x_n,a_n) =
  \begin{cases}
    x_n & \text{if } a_n = H.  \\
    x_n - 1 & \text{if} \text{ a platoon arrives and } a_n = D.   \\
    x_n - 1 + \kappa & \text{if} \text{ no platoon arrives and } a_n = D. \end{cases}
  \label{eq:cost}
\end{align}

We complete the specification of our Markov Decision Problem (MDP) by defining the transition probability function as follows:
\begin{align}
P(x_{n+1}|x_{n} = x, a_n) =
\begin{cases}
(1-p)(1-q) & \text{ if } x_{n+1} = Z(x), \hspace{.5cm} a_n = H.\\
(1-p)(1-q) & \text{ if } x_{n+1} = x - 1, \hspace{.42cm}a_n = D.\\
(1-p)q &  \text{ if } x_{n+1} = P(x), \hspace{.49cm} a_n = H.\\
(1-p)q &  \text{ if } x_{n+1} = x - 1, \hspace{.43cm}a_n = D.\\
p(1-q) & \text{ if }  x_{n+1} = V(x), \hspace{.49cm}a_n = H.\\
p(1-q) & \text{ if }  x_{n+1} = x, \hspace{1.1cm}a_n = D.\\
pq &  \text{ if }x_{n+1} = B(x), \hspace{.49cm}a_n = H.\\
pq &  \text{ if }x_{n+1} = x, \hspace{1.1cm}a_n = D.
\end{cases}
\label{transitions}
\end{align}

Our goal is to choose the control actions to minimize the expected finite horizon discounted cost, as 
\begin{align}
\mathbb{E} \sum_{n = 1}^{N} \beta^n c(x_n,a_n),
\label{eq:cost2}
\end{align}
where $\beta$ is a discount factor $0 < \beta < 1$, and $N$ is the time horizon.

Since $c(x_n,a_n)$ grows linearly $\forall x \in X$ and $N$ is finite, it is well known that there exists an optimal stationary (time-independent) policy that minimizes the cost \eqref{eq:cost2} and is the unique solution for the MDP. 

Define $J^\beta_N(x)$ as the minimum expected discounted cost for \eqref{eq:cost2} with $N$ steps to go and initial state $x_0 = x$. 
 $J^{\beta}_N(x)$ satisfies:
 \begin{align}
\notag &J^{\beta}_N(x) = \min\limits_{a\in \mathcal{A}(x)}\{ c(x,a) + \beta [ (1-p)(1-q)  {J^{\beta}_{N-1}(}Z(x{))} +    (1-p)q  {J^{\beta}_{N-1}(}P(x){)}\,+ \\ &\hspace{7.092cm}    p(1-q)  {J^{\beta}_{N-1}(}V(x){)} +    pq  {J^{\beta}_{N-1}(}B(x){)}] \}, \label{eq:dp2} 
 \end{align}
where the initial condition is $J^{\beta}_1(x) = \min \{c(x,a)| ~a\in \mathcal{A}(x)\}$. %, and $\tilde{a}(x)$ is a one-dimensional vector resulting from applying the action $a$ in a given state $x$. \ek{I AM BIT CONFUSED: Are $a$ and $\tilde{a}(x)$ representing the same thing? Why are there different notations?}

From the Dynamic Programming (DP) equation \eqref{eq:dp2}, we immediately see that with $n+1$ steps to go and initial state $x_0=x$, the optimal action $a$ is given by the difference $J_{n+1}^\beta(H(x)) - J_{n+1}^\beta(D(x))$, as:
\begin{align}
    a = \left\{\begin{array}{@{}lr@{}}
        \multirow{2.5}{*}{H, ~~if} & (1-p)(1-q)[{J^{\beta}_{n}}(x)-{J^{\beta}_{n}}(x-1)] + (1-p)q[{J^{\beta}_{n}}({x}) - {J^{\beta}_{n}}({x-1})]~~ \\
            &+\, p(1-q)[{J^{\beta}_{n}}(x+1)- {J^{\beta}_{n}}(x)] + pq[{J^{\beta}_{n}}({x+1})-{J^{\beta}_{n}}({x})] \leq \dfrac{-1+\kappa}{\beta}. \vspace{.cm}\\
                               \\
        \multirow{2.5}{*}{D, ~~if} & (1-p)(1-q)[{J^{\beta}_{n}}(x)-{J^{\beta}_{n}}(x-1)] + (1-p)q[{J^{\beta}_{n}}({x}) - {J^{\beta}_{n}}({x-1})]~~ \\
          &+\,p(1-q)[{J^{\beta}_{n}}(x+1)- {J^{\beta}_{n}}(x)] + pq[{J^{\beta}_{n}}({x+1})-{J^{\beta}_{n}}({x})] > \dfrac{-1+\kappa}{\beta}. \\
        \end{array}
        	%workaround
         \color{white}\right\} 
	\label{eq:cases}
\end{align}
We will use \eqref{eq:cases} to characterize the optimal policy for \eqref{eq:dp2} in the next section.

\section{Characterization of the Optimal Policy}\label{sec:optimal}

In this section, we first prove some properties of the optimal cost function; we then use them to characterize the optimal policy as a threshold policy.

The following Lemma (which we can show using Equations \eqref{eq:cost} and \eqref{eq:dp2})  simplifies the search for an optimal policy, since we can disregard the holding actions with probabilities $pq$ and $(1-p)q$ in Equation \eqref{transitions}. 
\begin{lemma}
\label{lemma:dispatch}
Dispatching a truck with an arriving platoon is always optimal. 
\end{lemma}

\begin{proof}
%The proof is immediate for the instantaneous cost \eqref{eq:cost} if we recall that when the system state is $x_n \geq 1$ and a platoon arrives, $c(x_n, D) < c(x_n, H)$ for all $x\in X$. We next show that the same follows for the expected discounted cost.
 
 We use induction on $n$ to prove that when a platoon arrives  
 \begin{align}
	J^{\beta}_n(D(x)) \leq J^{\beta}_n(H(x)), ~\forall x\in \text{dom} ~D.
	\label{eq:optimally}
\end{align} 

%We compare the instantaneous cost and the expected discounted cost of dispatching a vehicle when a platoon is present in contrast to holding it at the station.
 
\noindent \textbf{Base case ($n=1$).} This proof is immediate since $J^{\beta}_1(x)$ is the instantaneous cost \eqref{eq:cost} and when the system state is $x_n \geq 1$ and a platoon arrives, $c(x_n, D) < c(x_n, H)$ for all $x\in \text{dom} ~D$.  \\
%if we recall that $J^{\beta}_1 (x) = c(x, D)$ as $J^{\beta}_1 (x) = \min\{c(x,H), c(x, D)\}$, which is the instantaneous cost $\forall x\in \text{dom} ~D$. % when a platoon arrives and the system state is $x$, $x\geq 1$.

\noindent \textbf{Inductive step (from $n=N$ to $n=N+1$).} 
Assume that Inequality~\eqref{eq:optimally} is valid with $n=N$, for all $x \in \text{dom} ~D$. We will prove the same is true for all $x$ and $n=N+1$, i.e.,
\begin{align} 
J^{\beta}_{N+1}(D(x)) \leq J^{\beta}_{N+1}(H(x)).
\label{eq:eq:optimality0}
\end{align} 
\begin{comment}
We replace the operator $\tilde{a}$ by $H$ and $D$ in the corresponding equations and obtain  
\begin{align}
\notag	&c(x,D) + \beta [ (1-p)(1-q)  {J^{\beta}_{N}(}{D}(x){)} +    (1-p)q  {J^{\beta}_{N}(}{D}(x){)} \\ \notag
 &\hspace{2.1cm} +    p(1-q)  {J^{\beta}_{N}(}{D}(x+1){)} +    pq  {J^{\beta}_{N}(}{D}(x+1){)}] \leq \\
 \notag &\hspace{3.1cm} 	c(x,H) + \beta [ (1-p)(1-q)  {J^{\beta}_{N}(}{H}(x){)} +    (1-p)q  {J^{\beta}_{N}(}{H}(x){)} \\
 &\hspace{5.5cm} +    p(1-q)  {J^{\beta}_{N}(}{H}(x+1){)} +    pq  {J^{\beta}_{N}(}{H}(x+1){)}]. 
\end{align}
\end{comment}
We replace \eqref{eq:dp2} in \eqref{eq:eq:optimality0} and
 apply the operators $H$ and $D$ to obtain 
\begin{align}
\notag	x&-1 + \beta [ (1-p)(1-q)  {J^{\beta}_{N}(}x-1{)} +    (1-p)q  {J^{\beta}_{N}(}x-1{)}  \notag
 +    p(1-q)  {J^{\beta}_{N}}(x) +    pq  {J^{\beta}_{N}(}x{)}] \\
 \notag &\leq	x + \beta [ (1-p)(1-q)  {J^{\beta}_{N}}(x) +    (1-p)q  {J^{\beta}_{N}(}x{)}+    p(1-q)  {J^{\beta}_{N}(}x+1{)} +    pq  {J^{\beta}_{N}(}x+1{)}]. 
\end{align}
We now cancel identical terms and rearrange them to derive
\begin{align}
\notag	&-1 + \beta [ (1-p)(1-q) [ {J^{\beta}_{N}(}x-1{)} -{J^{\beta}_{N}}(x)]   +    (1-p)q  [{J^{\beta}_{N}(}x-1{)} -{J^{\beta}_{N}(}x{)}]\,+ \\
& \hspace{4.7cm} p(1-q) [ {J^{\beta}_{N}}(x) - {J^{\beta}_{N}}(x+1)]  +    pq  [ {J^{\beta}_{N}(}x{)} - {J^{\beta}_{N}(}x+1{)} ]]  \leq	\,0.
 \label{eq:optimality2}
\end{align}
\begin{comment}
Rearranging \eqref{eq:optimality2}, all we have to prove is that  
\begin{align}
	\frac{-1}{\beta} + (-J^{\beta}_N(x) + J^{\beta}_N(x-1)) + p(-J^{\beta}_N(x+1) -J^{\beta}_N(x-1) + 2J^{\beta}_N(x)) \leq 0.
	\label{eq:optimality3}
\end{align}
\end{comment}
From \eqref{eq:optimality2}, all we have to show is that $J^{\beta}_{N+1}(x)$ is monotone increasing for all $x$. However, since $J^{\beta}_N(x)$ is convex and $c(x_n,a_n)$ grows linearly $\forall x \in X$ (see Theorem \ref{lemma2}), monotonicity is sufficed. Therefore, Inequality~\eqref{eq:optimally} is valid for all $n \in N$ and  $x \in \text{dom} ~D$.
\end{proof} 

In Figure \ref{fig:transitionP}, we show the transition probabilities. 
\begin{figure}[h]
\centering
	\includegraphics[width=1\linewidth]{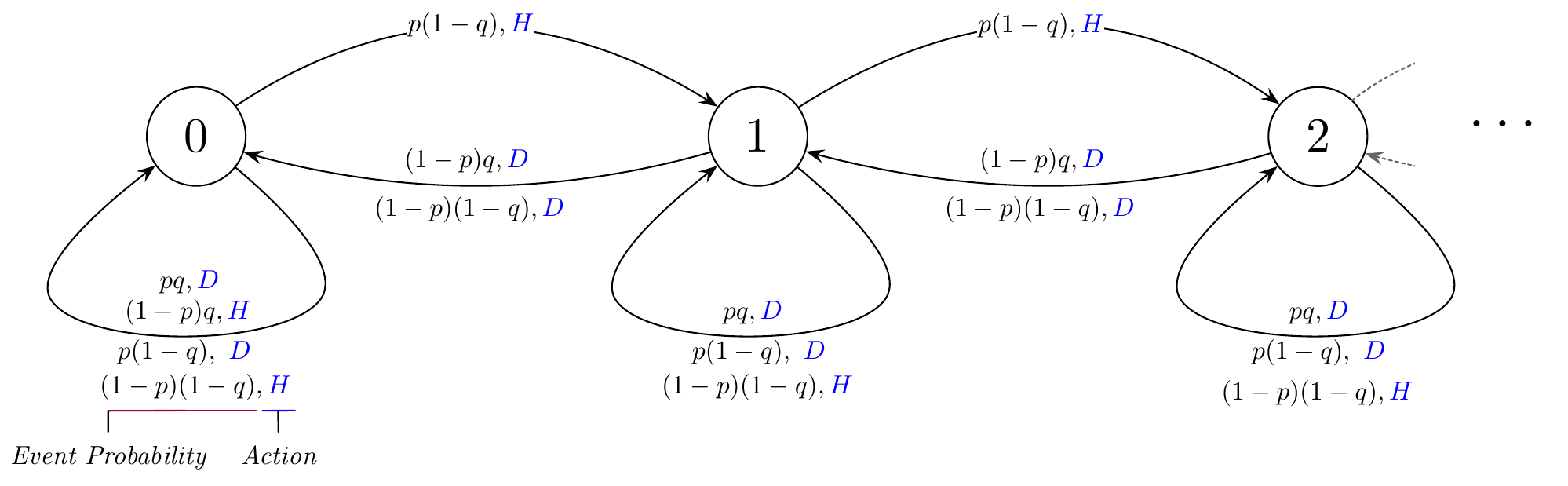}
	\caption{Transition probabilities using the result of Lemma \ref{lemma:dispatch}.}
	\label{fig:transitionP}
\end{figure}

\begin{theorem}
For each  $n\geq1$, the function $J^{\beta}_n$ is convex. Moreover, the difference $J_{n+1}^\beta(H(x)) - J_{n+1}^\beta(D(x))$ is monotone increasing.
\label{lemma2}
\end{theorem}
\begin{proof}
We first prove the convexity of $J^{\beta}_n$. We use induction on $n$ to prove that  for all $x$
\begin{align} 
\label{eq:convex1}
J^{\beta}_n(x+1) + J^{\beta}_n(x-1) \geq 2J^{\beta}_n(x) .
\end{align} 
 
\noindent \textbf{Base case ($n=1$).} This is immediate if we recall that $J^{\beta}_1 (x) = c(x,a)$, which is a linear function in $x$, for all $a$.

\noindent\textbf{Inductive step (from $n=N$ to $n=N+1$).} 
Assume that Inequality~\eqref{eq:convex1} is valid with $n=N$, for all $x$. We will prove the same is true for $n=N+1$, i.e.,
\begin{align} 
\label{eq:convex2}
J^{\beta}_{N+1}(x+1) + J^{\beta}_{N+1} (x-1) \geq 2J^{\beta}_{N+1} (x), 
\end{align} 
for all $x$ and $a \in \{H,D\}$. 

Remark: Using \eqref{eq:convex2}, we can see that  there are eight possible combinations of actions as $a \in \{H,D\}$.
 %we need to show that \eqref{eq:convex2} holds for all combinations of actions for $x-1$, $x$, and $x+1$, eight in total.
 However, from \eqref{eq:cases}, it follows that we need to prove that \eqref{eq:convex2} holds only for four of the cases, since once $D$ is optimal for $x$ (i.e., $x>m$), $D$ is optimal for all $x'>x$. \\

%In this section, we present the proof for one of the four cases. The proof for the remaining three cases is in the Appendix.

\noindent \textbf{Case 1: Holding action for $x-1$, $x$, and $x+1$.\\
}%\textbf{Let $a=H$ be optimal for $x-1$, $x$, and $x+1$.}
%(HHH) 
Applying the holding action to \eqref{eq:dp2}, the following identity holds
\begin{align}
J^{\beta}_{N+1} (x) 
&=  x + \beta 
( (1-p)(1-q) J^{\beta}_N (x) + p(1-q) J^{\beta}_N(x+1) \,+ (1-p) q J^{\beta}_N ({x}) + pq J^{\beta}_N( {x+1})).  \label{eq:convexH2-x}
%&= \notag x  +   \beta ( J^{\beta}_N(x) -qJ^{\beta}_N(x) -pJ^{\beta}_N(x) +2pqJ^{\beta}_N(x)  \\ & \hspace{1.3cm} +    pJ^{\beta}_N(x+1) -pqJ^{\beta}_N(x+1) + qJ^{\beta}_N(x-1) - pqJ^{\beta}_N(x-1))  \label{eq:convexH2-x}
\end{align} 
We apply Identity~\eqref{eq:convexH2-x} to $x-1$ and $x+1$ and we derive
\begin{align}
J^{\beta}_{N+1} (x+1) 
&= \label{eq:convexH2-x+1}  
 \notag x + 1 + \beta ( (1-p)(1-q) J^{\beta}_N (x+1) + p(1-q) J^{\beta}_N(x+2) + (1-p) q J^{\beta}_N ({x+1}) \,+ \\
& \hspace{11.5cm}  pq J^{\beta}_N( {x+2})).   \\
J^{\beta}_{N+1} (x-1) 
&= \label{eq:convexH2-x-1}  
 x - 1 + \beta ( (1-p)(1-q) J^{\beta}_N (x-1) + p(1-q) J^{\beta}_N(x) \,+  (1-p) q J^{\beta}_N ({x-1}) + pq J^{\beta}_N( {x})).
\end{align} 
Substituting Identities~\eqref{eq:convexH2-x},~\eqref{eq:convexH2-x+1},~and~\eqref{eq:convexH2-x-1} into Inequality~\eqref{eq:convex2} we obtain
\begin{align}
\label{eq:convexH3}
&  \notag  x + 1 + \beta ( (1-p)(1-q) J^{\beta}_N (x+1) + p(1-q) J^{\beta}_N(x+2) + (1-p) q J^{\beta}_N ({x+1}) + pq J^{\beta}_N( {x+2}))
\\  \notag & + x - 1 + \beta( (1-p)(1-q) J^{\beta}_N (x-1) + p(1-q) J^{\beta}_N(x) + (1-p) q J^{\beta}_N ({x-1}) + pq J^{\beta}_N( {x})) 
\\ & \hspace{.cm} \geq  2(x +\beta ( (1-p)(1-q) J^{\beta}_N (x) + p(1-q) J^{\beta}_N(x+1) + (1-p) q J^{\beta}_N ({x}) + pq J^{\beta}_N( {x+1}))).
\end{align}
Rearranging Inequality~\eqref{eq:convexH3} as indicated, all we need to prove is
\begin{align*}
&(1-p)(1-q)(J^{\beta}_N (x+1)+J^{\beta}_N (x-1)-2J^{\beta}_N (x)) + p(1-q)(J^{\beta}_N (x+2)+J^{\beta}_N (x)-2J^{\beta}_N (x+1)) + \\
&(1-p)q(J^{\beta}_N ({x+1})+J^{\beta}_N ({x-1})-2J^{\beta}_N ({x}))\,+ pq(J^{\beta}_N ({x+2})+J^{\beta}_N ({x})-2J^{\beta}_N ({x+1})) \geq 0.
\end{align*}
However, this is obvious using the convexity of $J^{\beta}_N$, which is true by the inductive hypothesis. \\

\noindent \textbf{Case 2: Dispatching action for $x-1$, $x$, and $x+1$.\\} 
Applying the dispatching action to \eqref{eq:dp2}, the following identity holds
\begin{align}
J^{\beta}_{N+1} (x) 
&= \label{eq:convex2-x}
  x-1+\kappa + \beta 
( (1-p)(1-q) J^{\beta}_N (x-1) + p(1-q) J^{\beta}_N(x) \,+ (1-p) q J^{\beta}_N ({x-1}) + pq J^{\beta}_N( {x})).  
\end{align} 
If we also apply Identity~\eqref{eq:convex2-x} to $x-1$ and $x+1$ and derive
\begin{align}
J^{\beta}_{N+1} (x+1) 
&= \label{eq:convex2-x+1}  
 x+\kappa + \beta 
( (1-p)(1-q) J^{\beta}_N (x) + p(1-q) J^{\beta}_N(x+1) \,+ (1-p) q J^{\beta}_N ({x}) + pq J^{\beta}_N( {x+1})). \\
J^{\beta}_{N+1} (x-1) 
&= \label{eq:convex2-x-1} 
\notag x-2+\kappa + \beta 
( (1-p)(1-q) J^{\beta}_N (x-2) + p(1-q) J^{\beta}_N(x-1) \,+ \\
& \hspace{8.3cm} (1-p) q J^{\beta}_N ({x-2}) + pq J^{\beta}_N( {x-1})).  
\end{align} 
Substituting Identities~\eqref{eq:convex2-x},~\eqref{eq:convex2-x-1},~and~\eqref{eq:convex2-x+1} into Inequality~\eqref{eq:convex2} we obtain
\begin{align}
& \notag x+\kappa + \beta 
( (1-p)(1-q) J^{\beta}_N (x) + p(1-q) J^{\beta}_N(x+1) \,+(1-p) q J^{\beta}_N ({x}) + pq J^{\beta}_N( {x+1}))\,+  \\
& \notag x-2+\kappa + \beta 
( (1-p)(1-q) J^{\beta}_N (x-2) + p(1-q) J^{\beta}_N(x-1) \,+ (1-p) q J^{\beta}_N ({x-2}) +  \\&pq J^{\beta}_N( {x-1})) +
\notag \geq 2(x-1+\kappa + \beta((1-p)(1-q) J^{\beta}_N (x-1) + p(1-q) J^{\beta}_N(x) \,+ \\ & \hspace{9cm}  (1-p) q J^{\beta}_N ({x-1}) + pq J^{\beta}_N( {x}))). 
\label{eq:convex3}
\end{align} 
Rearranging Inequality~\eqref{eq:convex3} as indicated, all we need to prove is
\begin{align*}
&(1-p)(1-q)(J^{\beta}_N (x)+J^{\beta}_N (x-2)-2J^{\beta}_N (x-1)) + p(1-q)(J^{\beta}_N (x+1)+J^{\beta}_N (x-1)-2J^{\beta}_N (x))  \\
&+(1-p)q(J^{\beta}_N ({x})+J^{\beta}_N ({x-2})-2J^{\beta}_N ({x-1}))\,+ pq(J^{\beta}_N ({x+1})+J^{\beta}_N ({x-1})-2J^{\beta}_N ({x})) \geq 0,
\end{align*}
which is true by the inductive hypothesis. \\

\noindent \textbf{Case 3: Dispatching action for $x$ and $x+1$, and holding action for $x-1$. \\} 
Substituting Identities    \eqref{eq:convexH2-x-1}, \eqref{eq:convex2-x}, and \eqref{eq:convex2-x+1}  into Inequality~\eqref{eq:convex2} we obtain
\begin{align}
& \notag x+\kappa + \beta 
( (1-p)(1-q) J^{\beta}_N (x) + p(1-q) J^{\beta}_N(x+1) \,+(1-p) q J^{\beta}_N ({x}) + pq J^{\beta}_N( {x+1}))\,+  \\
& \notag x - 1 + \beta( (1-p)(1-q) J^{\beta}_N (x-1) + p(1-q) J^{\beta}_N(x) + (1-p) q J^{\beta}_N ({x-1}) + pq J^{\beta}_N( {x})) \geq \\&
\notag 2(x-1+\kappa + \beta((1-p)(1-q) J^{\beta}_N (x-1) + p(1-q) J^{\beta}_N(x) \,+ (1-p) q J^{\beta}_N ({x-1}) + pq J^{\beta}_N( {x}))). 
\end{align}
Since \eqref{eq:convexH2-x-1} is identical to \eqref{eq:convex2-x}, we cancel equivalent  terms to derive 
\begin{align}
& \notag x+\kappa + \beta 
( (1-p)(1-q) J^{\beta}_N (x) + p(1-q) J^{\beta}_N(x+1) \,+(1-p) q J^{\beta}_N ({x}) + pq J^{\beta}_N( {x+1}))\geq x-1 \\&
 {+\,\kappa + \beta((1-p)(1-q) J^{\beta}_N (x-1) + p(1-q) J^{\beta}_N(x) \,+ (1-p) q J^{\beta}_N ({x-1}) + pq J^{\beta}_N( {x}))}. 
\label{eq:convexDHD3}
\end{align}
Rearranging Inequality~\eqref{eq:convexDHD3} as indicated, all we need to prove is
\begin{align}
& \notag (1-p)(1-q)(J^{\beta}_N (x) - J^{\beta}_N (x-1) ) +  p(1-q) (J^{\beta}_N(x+1) - J^{\beta}_N(x))\,+ \\  &\hspace{1.7cm} (1-p) q (J^{\beta}_N ({x}) - J^{\beta}_N ({x-1})) + pq (J^{\beta}_N( {x+1})-J^{\beta}_N( {x}))) \geq \frac{\kappa-1}{\beta}.
\label{eq:convexDHD4}
\end{align}
Since \eqref{eq:convexDHD4} is identical to \eqref{eq:cases}, the left side of \eqref{eq:convexDHD4} is greater than its right side whenever the dispatching action is taken. Thus, Inequality \eqref{eq:convexDHD4} holds. \\

\noindent \textbf{Case 4: Dispatching action for $x+1$, and holding action for $x-1$ and $x$. \\} 
Substituting Identities    \eqref{eq:convexH2-x}, \eqref{eq:convexH2-x-1}, and \eqref{eq:convex2-x+1}  into Inequality~\eqref{eq:convex2} we obtain
\begin{align}
& \notag x+\kappa + \beta 
( (1-p)(1-q) J^{\beta}_N (x) + p(1-q) J^{\beta}_N(x+1) \,+(1-p) q J^{\beta}_N ({x}) + pq J^{\beta}_N( {x+1}))\,+ 
\\  \notag &  x - 1 + \beta( (1-p)(1-q) J^{\beta}_N (x-1) + p(1-q) J^{\beta}_N(x) + (1-p) q J^{\beta}_N ({x-1}) + pq J^{\beta}_N( {x})) \geq
\\ & \hspace{.cm}   2(x +\beta ( (1-p)(1-q) J^{\beta}_N (x) + p(1-q) J^{\beta}_N(x+1) + (1-p) q J^{\beta}_N ({x}) + pq J^{\beta}_N( {x+1}))).
\end{align}
Since \eqref{eq:convexH2-x} is identical to \eqref{eq:convex2-x+1} (except for $\kappa$), we cancel equivalent terms to derive 
\begin{align}
 \notag &  \kappa +   x - 1 + \beta( (1-p)(1-q) J^{\beta}_N (x-1) + p(1-q) J^{\beta}_N(x) + (1-p) q J^{\beta}_N ({x-1}) + pq J^{\beta}_N( {x}))
\\ & \hspace{.cm}   \geq x +\beta ( (1-p)(1-q) J^{\beta}_N (x) + p(1-q) J^{\beta}_N(x+1) + (1-p) q J^{\beta}_N ({x}) + pq J^{\beta}_N( {x+1})).
\label{eq:convexDHHH}
\end{align}
Rearranging Inequality \eqref{eq:convexDHHH} as indicated, all we need to prove is
\begin{align}
& \notag (1-p)(1-q)(J^{\beta}_N (x-1) - J^{\beta}_N (x) ) +  p(1-q) (J^{\beta}_N(x) - J^{\beta}_N(x+1))\,+ \\  &\hspace{1.75cm} (1-p) q (J^{\beta}_N ({x-1}) - J^{\beta}_N ({x})) + pq (J^{\beta}_N( {x})-J^{\beta}_N( {x-1}))) \geq \frac{1-\kappa}{\beta}.
\label{eq:convexDHHH4}
\end{align}
Multiplying \eqref{eq:convexDHHH4} by $-1$, we obtain \eqref{eq:cases} when the holding action is chosen. Since both equations in \eqref{eq:convexDHHH4} take the holding action, \eqref{eq:convexDHHH4} holds by \eqref{eq:cases}. 

\textbf{Monotonicity.}

To prove that the difference $J_N^\beta(H(x))- J_N^\beta(D(x))$ is monotone increasing, we have to show that \begin{align}
\text{i.e.,~~~~} J_N^\beta(H(x))- J_N^\beta(D(x))\geq J_N^\beta(H(x-1))- J_N^\beta(D(x-1)), ~~~~~~~~\forall x \geq 1.
\label{eq:monotone}
\end{align}
Using Inequality \eqref{eq:cases}, we have to show that for every $x>m$, the following inequality holds\begin{align}
& \notag \dfrac{-\kappa+1}{\beta}+ (1-p)(1-q)({J^{\beta}_{n}}(x)-{J^{\beta}_{n}}(x-1)) + (1-p)q({J^{\beta}_{n}}({x}) - {J^{\beta}_{n}}({x-1}))\,  + \\\notag
     &\hspace{4.5cm}p(1-q)({J^{\beta}_{n}}(x+1)- {J^{\beta}_{n}}(x)) + pq({J^{\beta}_{n}}({x+1})-{J^{\beta}_{n}}({x})) \geq \\ &\dfrac{-\kappa+1}{\beta}+ (1-p)(1-q)({J^{\beta}_{n}}(m)-{J^{\beta}_{n}}(m-1)) + (1-p)q({J^{\beta}_{n}}({m}) - {J^{\beta}_{n}}({m-1}))\,\notag  + \\
     &\hspace{4cm}p(1-q)({J^{\beta}_{n}}(m+1)- {J^{\beta}_{n}}(m)) + pq({J^{\beta}_{n}}({m+1})-{J^{\beta}_{n}}({m})) > 0.
	\label{eq:sw3}
\end{align} We cancel identical terms in Inequality \eqref{eq:sw3} and rearrange it to derive
\begin{align}	
& \notag  (1-p)(1-q)({J^{\beta}_{n}}(x)-{J^{\beta}_{n}}(x-1) -(J^{\beta}_{n}(m) - J^{\beta}_{n}(m-1))) \,+\\\notag
     &\hspace{2cm}p(1-q)({J^{\beta}_{n}}(x+1)- {J^{\beta}_{n}}(x) - ({J^{\beta}_{n}}(m+1)- {J^{\beta}_{n}}(m)))\,+ \\&\hspace{3cm} (1-p)q({J^{\beta}_{n}}({x}) - {J^{\beta}_{n}}({x-1}) - ({J^{\beta}_{n}}({m}) - {J^{\beta}_{n}}({m-1})))\, \notag + \\&\hspace{4.6cm}   pq({J^{\beta}_{n}}({x+1})-{J^{\beta}_{n}}({x}) - ({J^{\beta}_{n}}({m+1})-{J^{\beta}_{n}}))> 0.
	\label{eq:sw4}
\end{align}
Since $J^\beta_N$ is convex and finite for all $x \in X$,  Inequality \eqref{eq:sw4} holds and so \eqref{eq:monotone}.
\end{proof}

From Equation~\ref{eq:cases} and this theorem, $\kappa \leq 1$ is a trivial case for which dispatching a truck (action $D$) is always optimal.

Denote by $\pi_m$  a \textit{threshold policy} with threshold $m$. Under policy $\pi_m$, the station dispatches a truck if and only if $x>m$ (with or without a platoon arrival). 

\begin{theorem}
%The policy which minimizes the discounted sum of costs of dispatching/holding vehicles over a finite horizon $N$ for the given MDP, is of threshold type.
	The optimal policy for the finite horizon discounted cost problem is of threshold type with a finite threshold.
\label{theorem:finite}
\end{theorem}

Consider next the infinite horizon, discounted cost.  This cost is defined as 
 $$J^\beta(x) = \lim_{n\rightarrow \infty} J_n^\beta(x).$$

Since $c(x_n,a_n)$ grows linearly $\forall x \in X$, it is also well known (e.g., see \cite{10.2307/2629409}) that an optimal stationary policy exists and  is the unique solution to the functional equation of dynamic programming.   
 \begin{theorem}
	The optimal policy for the infinite horizon discounted cost problem is  of threshold type.
	\label{theorem:infinite}
\end{theorem}
Note that the theorem does not imply that the threshold is finite; this would depend on (how small) the value of the discount factor $\beta$ is.

We extend Theorem~\ref{theorem:infinite} to the average cost case in the next section.

\section{The Average Cost problem} \label{sec:longrun}

\noindent In the average cost sense, we want to choose the dispatching actions so as to minimize the (long-run) average cost
\begin{align}
	\limsup\limits_{N\rightarrow \infty} \dfrac{1}{N} \mathop{\mathbb{E}}  \sum_{n=1}^{N} c(x_n,a_n).
 \label{eq:costaverage}
\end{align}     

Let $J_{\pi}^{\beta}(x)$, $J_{\pi}$ denote the infinite horizon, discounted cost and long-run average cost incurred by a policy $\pi$. 
It was shown in  \cite{10.2307/2629409}  that, if a policy $\pi$ results in a Markov chain with a single positive recurrent class, both costs are well defined and
\begin{align}
	\lim\limits_{\beta \rightarrow 1^{-}} (1-\beta) J_{\pi}^{\beta}(x) = J_{\pi},
	\label{eq:lippman}
\end{align} for any $x\in X$. 

%The following approach is similar to the one in \cite{1103637}, where we will show that the limit as $\beta\rightarrow 1^-$ of a sequence of threshold policies is a threshold policy again.

Define $J_{\pi_m}$ as the average cost of using the threshold policy $\pi_m$. We follow an approach  similar to the one in \cite{1103637}, to prove the optimality of $\pi_m$ for the average cost criterion.

The following lemma shows that this cost is well-defined and finite. 
\begin{lemma}
	For any finite threshold $m\geq 0$, the average cost of the policy $\pi_m$ is finite and given by:
	\label{lemma:m>4}
\begin{align}
J_{\pi_m} =
	\begin{cases}
		   p(1-q) \kappa. ~~~~\text{ if } m = 0. \\ \\
		 \dfrac{p(1-q)}{A+1}+ \dfrac{A(p(1-q)(1+\kappa) +(1-p)(1-q) + pq)}{A+1}, ~~~~\text{ if } m = 1. \\ \\
		  \dfrac{m^2 + m - 2 (\kappa + 1) (p - 1) p}{2 (m + 1)}, ~~\text{ if } m \geq 2 \text{ and } p=q. \\ \\
		 \dfrac{p(1-q)(1-A)}{1-A^{m+1}}+ \dfrac{A^2 (-p + q) + A(1 + p - q) + A^m(-m-p+q) }{(1-A)(1-A^{m+1})} +\vspace{0.2cm} \\ \dfrac{A^{m+1}(m+p-q-1)}{(1-A)(1-A^{m+1})} +(m-q + pq(1-\kappa)+p\kappa)\Big(\dfrac{A^{m} -A^{m+1}}{1-A^{m+1}}\Big),\text{ if } m \geq 2 %\\ \hspace{10.695cm} 
		  \text{ and } p\neq q.
\end{cases}
\label{eq:finalJm2}
\end{align} 
\begin{comment}
For $p= q$,
\begin{align}
J_{\pi_m} =
	\begin{cases}
		   -p^2\kappa +p\kappa. ~~~~\text{ if } m = 0. \\ \\
		 \dfrac{p(1-q)}{2}+ \dfrac{p(1-q)\kappa +(1-p)(1-q) + pq}{2}. ~~~~\text{ if } m = 1. \\ \\
		% \dfrac{m^2 -m}{2(m+1)} +\dfrac{[m+p^2-p+p\kappa-p^2\kappa]}{m+1}. 
		\dfrac{m^2 + m - 2 p^2\kappa + 2p\kappa}{2(m + 1)} ~~\text{ if } m \geq 2.
\end{cases} \hspace{4.4cm}
\label{eq:finalJm3}
\end{align} 
\end{comment}
where \begin{align}
	&A= \frac{p(1-q)}{(1-p)q}. 
\label{eq:f0q1}
 \end{align}
\begin{comment}
	 \begin{align}
	&A= \frac{p^2}{q^2} & 
	f(0) =
\begin{cases}
 	1 -\bigg(\dfrac{Af(0) (1- A^m)}{1-A}\bigg) ~~~~\text{if}~~p\neq q. \\
	\dfrac{1}{m+1}. ~~~~\text{if}~~p= q. 
\end{cases}
\label{eq:f0q1}
 \end{align}
\end{comment}
%and $f(0)$ is the steady-state probability of the state $x = 0$.
\end{lemma}
%\noindent\textit{Proof.}	In the Appendix. \\

\begin{proof}
Let $f(x)$ be the steady-state probability of the state $x\in X$ when $\pi_m$ is applied. ($\pi_m$ has a unique steady-state distribution since, under $\pi_m$, (i) the queue will never drift to the infinity and (ii) the state $x=0$ can be reached from every state in $X$.) $f(x)$ is defined as:  \begin{align}
	 f(x) &= (1-p)(1-q) f(\pi_m(Z(x))) + (1-p)qf(\pi_m(P(x)))+ p(1-q)f(\pi_m(V(x))) + pqf(\pi_m(B(x))). \label{eq:steadydef}
\end{align}
The state transition diagram using policy $\pi_m$ and the definition \eqref{eq:steadydef} is shown in Figure \ref{fig:stateT}.
\begin{figure}[h]
\begin{subfigure}{1\textwidth}
  \centering
  \includegraphics[width=.45\linewidth]{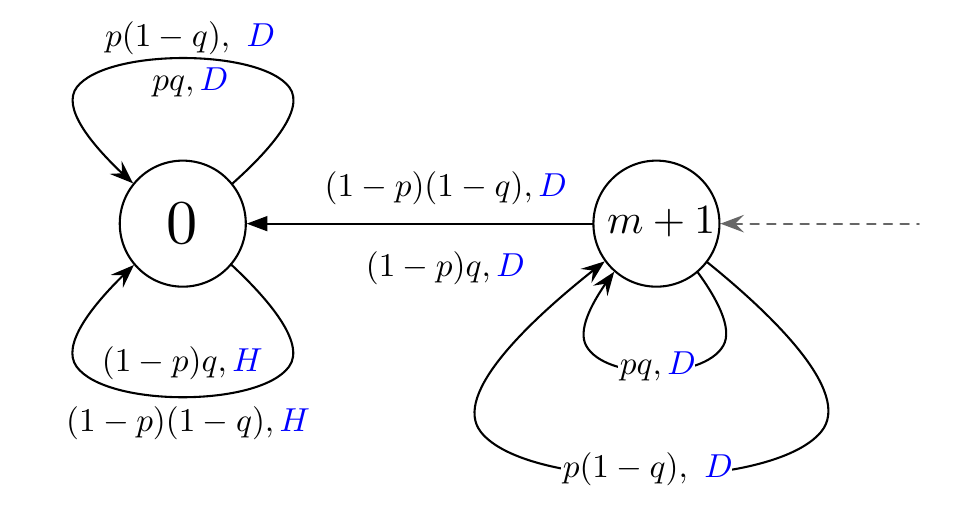}
  \caption{For $m = 0$.}
  \hfill
\end{subfigure}

\begin{subfigure}{1\textwidth}
  \centering
  \includegraphics[width=.65\linewidth]{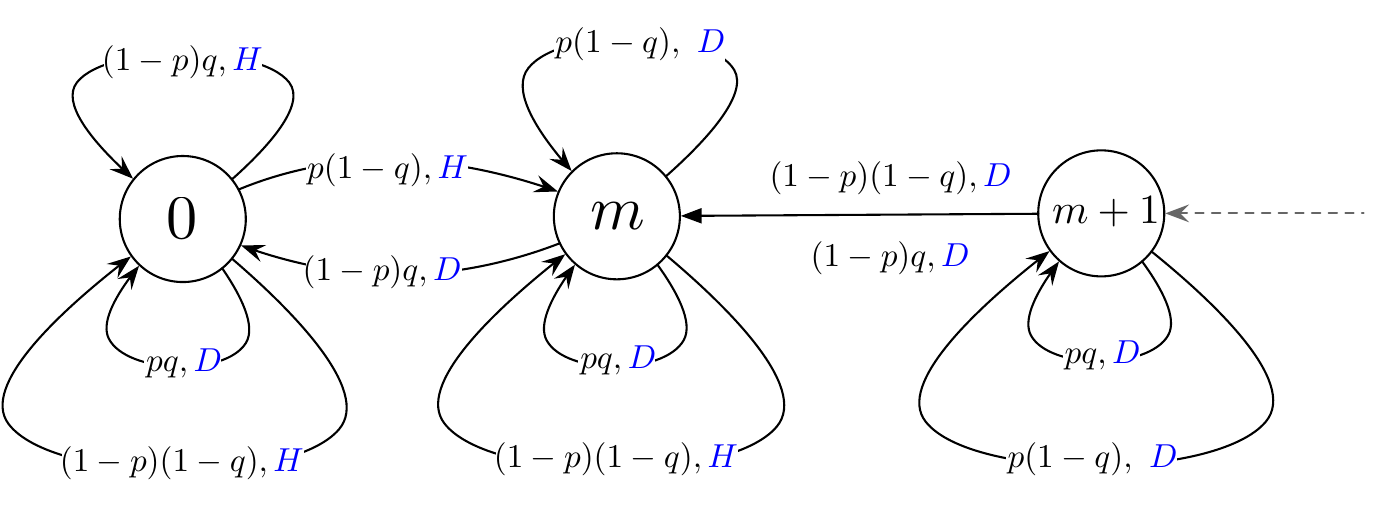}
  \caption{For $m = 1$.}
  \hfill
\end{subfigure}%

\begin{subfigure}{1\textwidth}
  \centering
  \includegraphics[width=1\linewidth]{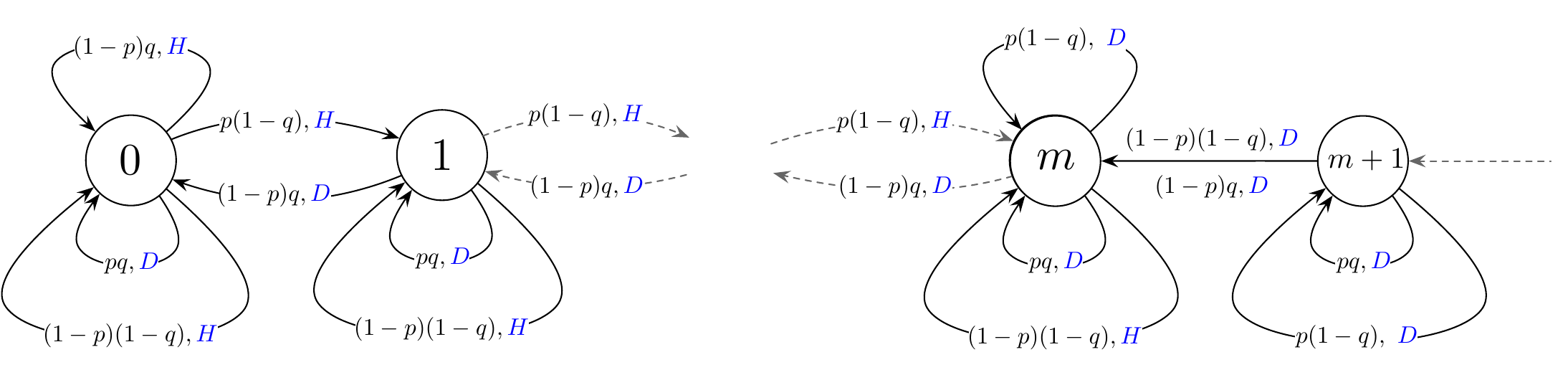}
\caption{For $m\geq 2$.}
\end{subfigure}
\caption{State transition diagram resulting from a threshold policy with threshold $m$.}
\label{fig:stateT}
\end{figure}
\begin{comment}
\begin{figure}[H]
  \includegraphics[width=\textwidth]{figures/overview/state_transitionsf2.pdf}
  \caption{State transition diagram resulting from a threshold policy with threshold $m \geq 2$.}
  \label{fig:stateT}
\end{figure}

\begin{figure}[H]
\centering
  \includegraphics[width=.5\textwidth]{figures/overview/state_transitionsm=0.pdf}
  \caption{State transition diagram resulting from a threshold policy with threshold $m$.}
  \label{fig:stateT}
\end{figure}
\end{comment}

Using the definition of $\pi_m$ and \eqref{eq:steadydef}, we derive for $m\geq 2$ %$f(0)$, $f(x)$ for all $x \in[1,t-2]$, $f(t-1)$, and $f(x')$ for all $x \geq t$
\begin{align}
\text{if ~}	x=0,\hspace{1.65cm}&f(0) = ((1-p)q+pq+(1-p)(1-q))f(0) + ((1-p)q)f(1). \label{eq:p0} \\
\text{if ~}	x\in[1,m-1],\hspace{.42cm}&\notag f(x) = (p(1-q))f(x-1) + ((1-p)(1-q)+pq)f(x) \\&\hspace{4.25cm}+((1-p)q + (1-p)(1-q))f(x+1). \label{eq:pt-20} \\
	\text{if ~}	x=m,\hspace{1.51cm}&\notag f(m) = (p(1-q))f(m-1) + ((1-p)(1-q)+pq+p(1-q))f(m) \\ &\hspace{4.45cm}+((1-p)q + (1-p)(1-q))f(m+1). \label{eq:pt-1} \\
	\text{if ~}	x> m,\hspace{1.52cm}& f(x) = (p(1-q) + pq)f(x) + ((1-p)q + (1-p)(1-q))f(x+1) \label{eq:px}.
\end{align}
%where we use \eqref{eq:p0} for $x=0$, \eqref{eq:pt-20} for all $x \in[1,t-2]$, \eqref{eq:pt-1} for $x=t-1$,  and \eqref{eq:px} for all $x = x'\geq t$. 
It is more convenient to rewrite \eqref{eq:p0}, \eqref{eq:pt-20},  and \eqref{eq:pt-1} as:% (We use the fact that the sum of the probabilities is unity.)
\begin{align}
	&f(0) = (1 - p(1-q))f(0) + ((1-p)q)f(1). \label{eq:p01} \\
	&f(x) = (p(1-q))f(x-1) + (1 - p(1-q) - (1-p)q)f(x)+((1-p)q)f(x+1). \label{eq:pt-2} \\
	&f(m) = (p(1-q))f(m-1) + (1-q+pq)f(m) + (1-p)f(m+1). \label{eq:pt-12} %\\
	%&f(m+1) = \frac{-(p(1-q))f(m-1) + [p(1-q)+(1-p)q]f(m)}{1-p}
\end{align}
Using \eqref{eq:p01}, we deduce $f(1)$ as:
\begin{align}
	\notag f(1) &= \frac{-(1 - p(1-q))f(0) + f(0)}{(1-p)q} \\
	\notag &= \frac{-f(0) + (p(1-q))f(0) + f(0)}{(1-p)q} \\
	 &= \bigg(\frac{p(1-q)}{(1-p)q}\bigg)f(0). \label{eq:p1} 
\end{align}
From \eqref{eq:pt-2} and \eqref{eq:p1}, we deduce $f(2)$:
\begin{align}
	f(1) \notag &= (p(1-q))f(0) + (1 - p(1-q) - (1-p)q)f(1)+((1-p)q)f(2). \\
	f(2) &= \frac{f(1) -(p(1-q))f(0) + (-1 + p(1-q) + (1-p)q)f(1)}{(1-p)q}.  \label{eq:p2}
\end{align}
\begin{comment}
We rewrite the probabilities $(1-p)(1-q)+pq$ as $1-p-q+pq$ and derive
\begin{align}
\notag		f(2) &= \frac{f(1) -(p(1-q))f(0) - [1-p-q+pq+pq]f(1)}{(1-p)q}. 
\end{align}
\end{comment}
We cancel the terms $-f(1) + f(1)$ to derive
\begin{align}
\notag	f(2)&= \frac{f(1) -f(1)  -(p(1-q))f(0) + (p(1-q) + (1-p)q)f(1)}{(1-p)q}. \\
\notag	&= -\bigg(\frac{p(1-q)}{(1-p)q}\bigg)f(0) + \bigg(\frac{p(1-q) + (1-p)q}{(1-p)q}\bigg)f(1). \\
\notag	&= -\bigg(\frac{p(1-q)}{(1-p)q}\bigg)f(0) + \bigg(\frac{p(1-q)}{(1-p)q}+{1}\bigg)f(1)
%&p( 1-q) + (1-p)q = p -pq + q - pq -\rightarrow \frac{p(1-q) + (1-p)q}{(1-p)q} \rightarrow \\ 
%&\frac{p(1-q)}{(1-p)q} + \frac{(1-p)q}{(1-p)q} = &\frac{p(1-q)}{(1-p)q}f(1) + 1f(1)
	%\bigg(\frac{p(1-q)}{(1-p)q}\bigg)f(0) &= (p(1-q))f(0) + ((1-p)(1-q)+pq)\bigg(\frac{p(1-q)}{(1-p)q}\bigg)f(0) +
\end{align}
Replacing $f(1)$ by \eqref{eq:p1}, we obtain:
\begin{align}
%\notag		f(2) &= \frac{f(1) -(p(1-q))f(0) - [1-p-q+2pq]f(1)}{(1-p)q}. \\
%\notag	&= \frac{f(1) -f(1)  -(p(1-q))f(0) - [-p-q+2pq]f(1)}{(1-p)q}. \\
	\notag f(2)&= -\bigg(\frac{p(1-q)}{(1-p)q}\bigg)f(0) + \bigg(\frac{p(1-q)}{(1-p)q}+{1}\bigg)\bigg(\frac{p(1-q)}{(1-p)q}\bigg)f(0).   \\
	&=  \bigg(\frac{p(1-q)}{(1-p)q}\bigg)\bigg(\frac{p(1-q)}{(1-p)q}\bigg)f(0). \label{eq:p22}
	%\bigg(\frac{p(1-q)}{(1-p)q}\bigg)f(0) &= (p(1-q))f(0) + ((1-p)(1-q)+pq)\bigg(\frac{p(1-q)}{(1-p)q}\bigg)f(0) +
\end{align}
Define $A$ as
\begin{align}
A &= \frac{p(1-q)}{(1-p)q}.
%B &= \frac{p+q-2pq}{(1-p)q}.
\end{align}
We replace $A$ in \eqref{eq:p22} (as it will repeat quite often)
\begin{align}
%\notag		f(2) &= \frac{f(1) -(p(1-q))f(0) - [1-p-q+2pq]f(1)}{(1-p)q}. \\
%\notag	&= \frac{f(1) -f(1)  -(p(1-q))f(0) - [-p-q+2pq]f(1)}{(1-p)q}. \\
	f(2)&= A^2f(0).  \label{eq:p23}
	%\bigg(\frac{p(1-q)}{(1-p)q}\bigg)f(0) &= (p(1-q))f(0) + ((1-p)(1-q)+pq)\bigg(\frac{p(1-q)}{(1-p)q}\bigg)f(0) +
\end{align}
%or the equivalent 
%\begin{align}
% f(2)&= -A[f(0)] + B[f(1)].
% \end{align}
By induction (see Lemma \ref{lemma:fx}), we rewrite \eqref{eq:pt-2} for all $x \in [1,m]$ as:
\begin{align}
	f(x)&= A^xf(0).  \label{eq:pt-21}
\end{align}
%This result holds for all $x\in[1,m]$. 
We now derive the result for $x>m$.
We can simplify \eqref{eq:px} as
\begin{align}
	\notag f(x) &= (p(1-q) + pq)f(x) + ((1-p)q + (1-p)(1-q))f(x+1). \\
	\notag &= (p-pq+pq)f(x)+(q-pq  + 1-p-q+pq)f(x+1). \\
	\notag &= pf(x) + (1-p)f(x+1). \\
	\notag (1-p)f(x) &= (1-p)f(x+1). \\
	f(x) &= f(x+1). \label{eq:px-sim}
\end{align} 
The result in \eqref{eq:px-sim}
 holds for all $x > m$. We now derive $f(m)$ where $x=m$ using \eqref{eq:pt-12}, we also solve \eqref{eq:px-sim}.
\begin{align}
	\notag f(m) &= (p(1-q))f(m-1) + (1-q+pq)f(m) + (1-p)f(m+1). \\
	\notag f(m+1)&= \frac{-(p(1-q))f(m-1) + ((1-p)q)f(m)}{1-p}. \\
	\notag &= -\frac{p(1-q)}{1-p}f(m-1) + qf(m). \\
	  \notag &= -\frac{p(1-q)}{1-p}A^{m-1}f(0) + q\bigg(\frac{p(1-q)}{(1-p)q}\bigg) A^{m-1}f(0). \\
	  \notag &= -\frac{p(1-q)}{1-p}A^{m-1}f(0) + \frac{p(1-q)}{(1-p)} A^{m-1}f(0). \\
	 &=0 \label{eq:pxderived}
\end{align}
The result in \eqref{eq:pxderived} holds for all $x > m$. We collect all the results as follows
\begin{align}
	%\text{ if }& x = 0, \hspace{1.9cm}  f(0)= \bigg[\frac{(1-p)q}{p(1-q)}\bigg]f(1). \\
	\text{ if }& 1\leq x \leq m, \hspace{1.19cm} f(x)= A^xf(0). \\
	\text{ if }& x> m,  \hspace{1.83cm}  f(x)= 0.
\end{align}
We obtain $f(0)$ from the condition that the sum of the probabilities is unity.
\begin{comment}
\begin{align}
f(0) =
\begin{cases}
 	1 - \Big\{\sum^{m}_{x=1} A^x f(0)\Big\} ~~~~\text{if}~~p\neq q. \\
 	1 - \Big\{\sum^{m}_{x=1} f(0)\Big\} ~~~~\text{if}~~p= q. \\
\end{cases}
\end{align} 
\end{comment}
\begin{align}
f(0) =
\begin{cases}
 	\dfrac{A-1}{A^{m+1}-1} ~~~~\text{if}~~p\neq q. \\ \\
	\dfrac{1}{m+1}. ~~~~\text{if}~~p= q. 
\end{cases}
\label{eq:f0q}
\end{align} 
%\begin{align}
%	f(1) = 1 - [Df(0) + Cp[f(0)]].
%\end{align}

Let the $J_{\pi_m}$ be the cost of using policy $\pi_m$ to govern the platooning system. 
\begin{align}
J_{\pi_m} = \sum_{x} c(x,a)f(x). \label{eq:steadycost}
\end{align}
Since $c(x,a) \in \{x,x-1,x-1+\kappa\}$ for all $x \in\{1,m\}$, we obtain from Equation \eqref{eq:steadycost}
\begin{align}
J_{\pi_m} =
	\begin{cases}
		\notag  p(1-q)f(0) + \bigg\{ \sum\limits^{m-1}_{x\,=\,1} (p - q + x)A^xf(0) \bigg\} \,+ \\ -A^m(q(m - 1)(p - 1) - m(pq + (p - 1)(q - 1)) + p(\kappa + m)(q - 1))f(0). ~~~~\text{ if }p \neq q.
\\
\\
p(1-p)f(0) + \bigg\{ \sum\limits^{m-1}_{x\,=\,1} xA^xf(0) \bigg\} \,+ \\ -A^m(p(m - 1)(p - 1) - m(p^2 + (p - 1)^2) + p(\kappa + m)(p - 1))f(0). ~~~~\text{ if }p = q.   	
\end{cases}
\end{align}
\begin{comment}
\begin{align}
	&p(1-q)f(0) + \bigg\{ \sum^{m-1}_{x=1} (p - q + x)A^xf(0) \bigg\} \\ &-A^m(q(m - 1)(p - 1) - m(pq + (p - 1)(q - 1)) + p(\kappa + m)(q - 1))f(0)
\end{align}
\end{comment}

Or the equivalent after some algebraic manipulation and simplification
\begin{align}
J_{\pi_m} =
\begin{cases}
\displaystyle              %%5 <--- here
\dfrac{p(1-q)(1-A)}{1-A^{m+1}}+ \dfrac{A^2 (-p + q) + A(1 + p - q) + A^m(-m-p+q) }{(1-A)(1-A^{m+1})} +\vspace{0.2cm} \\ \dfrac{A^{m+1}(m+p-q-1)}{(1-A)(1-A^{m+1})} +(m-q + pq(1-\kappa)+p\kappa)\Big(\dfrac{A^{m} -A^{m+1}}{1-A^{m+1}}\Big). ~~~\text{ if }p \neq q.\\ \\
 \dfrac{m^2 + m - 2 (\kappa + 1) (p - 1) p}{2 (m + 1)}. ~~~\text{ if }p = q.
	\end{cases}
	\label{eq:finalJm}
\end{align}

%\subsubsection{For $m = 0$}
For $m=0$,
\noindent we use \eqref{eq:steadydef} to derive 
\begin{align}
\text{if ~}	x=0,\hspace{1.65cm}&f(0) = 1. \label{eq:pm=00} \\
	\text{if ~}	x> m,\hspace{1.52cm}& f(x) = 0 \label{eq:pm=0x}.
\end{align}
Then, $J_{\pi_0}$ is calculate as follows
\begin{align}
J_{\pi_0} = p(1-q) \kappa. \end{align}

For $m=1$, we use \eqref{eq:steadydef} to derive 
\begin{align}
\text{if ~}	x=0,\hspace{1.65cm}&f(0) = (1 - p(1-q))f(0) + ((1-p)q)f(m). \label{eq:pm=10} \\
	\text{if ~}	x=m,\hspace{1.51cm}& f(m) = (p(1-q))f(0) + (1 - (1-p)q)f(m) +(1-p)f(m+1). \label{eq:pm=1t-1} \\
	\text{if ~}	x> m,\hspace{1.52cm}& f(x) = (p(1-q) + pq)f(x) + ((1-p)q + (1-p)(1-q))f(x+1) \label{eq:pm=1x}.
\end{align}
We deduce $f(m)$ using similar calculations to \eqref{eq:p01} as:
\begin{align}
	f(m) = \frac{p(1-q)}{(1-p)q} f(0). \label{eq:pm=11} 
\end{align}
It is easy to see that \eqref{eq:pm=1t-1} is similar to \eqref{eq:pt-21} so,
\begin{align}
	f(m) = Af(0).
	\label{eq:pm=1t-21}
\end{align}
Finally, for $x>m$, we the same calculations as in  \eqref{eq:pxderived}. 
We collect all the results as follows
\begin{align}
	%\text{ if }& x = 0, \hspace{1.9cm}  f(0)= \bigg[\frac{(1-p)q}{p(1-q)}\bigg]f(1). \\
	\text{ if }& x = m, \hspace{1.19cm} f(x)= Af(0). \\
	\text{ if }& x> m,  \hspace{1.18cm}  f(x)= 0.
\end{align}
We again obtain $f(0)$ from the condition that the sum of the probabilities is unity.
\begin{align*}
	f(0) = \dfrac{1}{A+1}.%\label{eq:f0q}
\end{align*}
Then, $J_{\pi_1}$ is calculate as follows
\begin{align}
	J_{\pi_1} =\dfrac{p(1-q)}{A+1}+ \dfrac{A(p(1-q)(1+\kappa) +(1-p)(1-q) + pq)}{A+1}.
\end{align}
\end{proof}

\begin{lemma}
 $f(x) = A^x f(0),$ for all $x \in [1,m-1]$.
 \label{lemma:fx}
\end{lemma}
\begin{proof} The proof is by induction. 

\textbf{Base case ($x=1$).} This proof is immediate by \eqref{eq:p1}.	

\textbf{Inductive step (from $x=m-2$ to $x=m-1$).}  Assume that  \begin{align}
	f(m-2) = A^{m-2} f(0),\label{eq:pm-2}
\end{align} is valid for $x=m-2$. We will prove the same is true for $x=m-1$, i.e., $$f(m-1) = A^{m-1} f(0).$$
From \eqref{eq:pm-2}, we have that, for $x=m-2$
	\begin{align}
		 f(m-2) &= (p(1-q))f(m-3) + (1 - p(1-q) - (1-p)q)f(m-2)+((1-p)q)f(m-1).  
	\end{align}
Using similar arguments as in \eqref{eq:p22}, we obtain
	\begin{align*}
		f(m-1)& = -\frac{p(1-q)}{(1-p)q}f(m-3) + \frac{(p(1-q) + (1-p)q)}{((1-p)q)}f(m-2). \\
		  &=  -\frac{p(1-q)}{(1-p)q}f(m-3) + \frac{p(1-q)}{(1-p)q}f(m-2) + \frac{(1-p)q}{(1-p)q}f(m-2) \\
		  &= -\frac{p(1-q)}{(1-p)q}f(m-3) + f(m-2) + \frac{p(1-q)}{(1-p)q}f(m-2)
	\end{align*}
Substituting $f(m-2)$ and  $f(m-3)$, we obtain 
\begin{align}
		\notag f(m-1)& = -\frac{(p(1-q))}{((1-p)q)}A^{m-3}f(0) + A^{m-2}f(0) + \frac{p(1-q)}{(1-p)q}A^{m-2}f(0) \\
		\notag &= -AA^{m-3}f(0) + A^{m-2}f(0) + \frac{p(1-q)}{(1-p)q}A^{m-2}f(0).\\
		\notag  &= AA^{m-2}f(0). \\
		 &= A^{m-1}f(0). \label{eq:finalrlemma}
	\end{align}
	From \eqref{eq:finalrlemma} we see that $f(x)=A^xf(0)$ for all $x \in [1,m-1]$. 
\end{proof}

We next investigate the asymptotic behaviour of \eqref{eq:finalJm2}.
\begin{lemma}
	For $p\geq q$, it follows from \eqref{eq:finalJm2} that 	\begin{align}
		\lim_{m\rightarrow\infty} J_{\pi_m} = \infty.
		\label{eq:asymptotic1}
	\end{align}
	For $p<q$, it follows from \eqref{eq:finalJm2} that
	\begin{align}
		&\lim\limits_{m\rightarrow\infty} J_{\pi_m} = p(1-q)(1-A) + \frac{A^2 (-p + q) + A (1 + p - q)}{(1-A)} %+\\ &\hspace{6cm} \frac{A^{m+2} (m + p - q) + A^{m+1} (-m - p + q-1)}{(1-A)} 
		\label{eq:asymptotic2}
	\end{align}
	\label{lemma:asymptotic}
\end{lemma}
\begin{proof}
	For $p\geq q$, $A>1$ so the term $mA^{m+1}$ in \eqref{eq:finalJm2} as $m\rightarrow\infty$ will be sufficiently larger than $m$. For $p<q$, $A<1$ so the terms  $A^{m}$ and $A^{m+1}$ in \eqref{eq:finalJm2} as $m\rightarrow\infty$ will become insignificant.
\end{proof}

\begin{lemma}
	There exists a finite $m$ such that $J_{\pi_m}< J_{\pi_{m+1}}$.
	\label{lemma:m}
\end{lemma}
\begin{proof}
For $p \geq q$, the result follows by Lemma \ref{lemma:asymptotic}. For $p < q$, we form the  difference $J_{\pi_{m+1}}-J_{\pi_{m}}$ and find
\begin{comment}
\begin{align}
	\notag &J_{\pi_{m+1}}-J_{\pi_{m}}= \Bigg[\frac{p(1-q)(A-1)}{A^{m+2}-1}+ \\ \notag & \hspace{2.8cm} \frac{A^2 (-p + q) + A(1 + p - q)}{(A -1)(A^{m+2}-1)} + \frac{A^{m+3} (m + p - q) + A^{m+2} (-m - p + q-1)}{(A -1)(A^{m+2}-1)} + \\\notag &  \hspace{2cm}\Big[(m+1)-p-q+2pq+p\kappa-pq\kappa\Big]\Big(\dfrac{A^{m+2} -A^{m+1}}{A^{m+2}-1}\Big)\Bigg] - \Bigg[\frac{p(1-q)(A-1)}{A^{m+1}-1}+ \\ \notag & \hspace{2cm} \frac{A^2 (-p + q) + A(1 + p - q)}{(A -1)(A^{m+1}-1)} + \frac{A^{m+2} (m + p - q) + A^{m+1} (-m - p + q-1)}{(A -1)(A^{m+1}-1)} + \\\notag &  \hspace{2cm}\Big[m-p-q+2pq+p\kappa-pq\kappa\Big]\Big(\dfrac{A^{m+1} -A^m}{A^{m+1}-1}\Big)\Bigg]. 
	\label{eq:thedifferences0}
\end{align}
\end{comment}
\begin{align}
	\notag &J_{\pi_{m+1}}-J_{\pi_{m}}= \Bigg[\frac{B}{1-A^{m+2}}+  \frac{C + A^{m+1}(-(m+1)-p+q) + A^{m+2}((m+1)+p-q-1)}{(1-A)(1-A^{m+2})}  + \\ \notag &   \hspace{5cm}  (D+m+1)\Bigg(\dfrac{A^{m+1} -A^{m+2}}{1-A^{m+2}}\Bigg)\Bigg]\,- \\ \notag &\hspace{2.6cm} \Bigg[\frac{B}{1-A^{m+1}}+  \frac{C + A^{m}(-m-p+q) + A^{m+1}(m+p-q-1)}{(1-A)(1-A^{m+1})} + \\  & \hspace{5cm} (D+m)\Bigg(\dfrac{A^{m} -A^{m+1} }{1-A^{m+1}}\Bigg)\Bigg], 
	\label{eq:thedifferences0}
\end{align}
where $B$, $C$, $(D+m)$, and $(D+m+1)$  are positive, and
\begin{align*}
	&B = p(1-q)(1-A). &~~~~~~C=A^2 (-p + q) + A(1 + p - q). &~~~~~~D=-q + pq(1-\kappa)+p\kappa.
\end{align*}
%\begin{align*}
%	&E = {A^{m+2} (m + p - q) + A^{m+1} (-(m +1) - p + q)} &F={A^{m+1} (m -1 +  p - q) + %A^{m} (-m - p + q)}
%\end{align*}
We rewrite \eqref{eq:thedifferences0} by grouping identical elements as
\begin{align}
\notag	&J_{\pi_{m+1}}-J_{\pi_{m}}= \\&\Bigg(\frac{A^{m+2}((m+1)+p-q-1)}{(1-A)(1-A^{m+2})} - \frac{A^{m+1}(m+p-q-1)}{(1-A)(1-A^{m+1})} \Bigg) + \notag \textcolor{black}{B}\Bigg(\frac{1}{1-A^{m+2}} - \frac{1}{1-A^{m+1}}\Bigg) \\& \notag  \Bigg(\frac{A^{m+1}(-(m+1)-p+q)}{(1-A)(1-A^{m+2})}  - \frac{A^{m}(-m+-p+q)}{(1-A)(1-A^{m+1})} \Bigg)  +\textcolor{black}{C}\Bigg(\frac{1}{(1 -A)(1-A^{m+2})} - \frac{1}{(1-A)(1-A^{m+1})}\Bigg) \notag \\
	&\hspace{0cm}\Bigg(\dfrac{ (m+1)(A^{m+2} -A^{m+1})}{1-A^{m+2}} - \dfrac{ m(A^{m+1} -A^{m})}{1-A^{m+1}}\Bigg) + \textcolor{black}{D}\Bigg(\dfrac{A^{m+2} -A^{m+1}}{1-A^{m+2}}-\dfrac{A^{m+1} -A^m}{1-A^{m+1}}\Bigg).
\label{eq:thedifferences01}
\end{align}
From \eqref{eq:thedifferences01}, it follows that the multiplication of the constant terms ($B$, $C$, and $D$) will be negative. We then rewrite \eqref{eq:thedifferences01}, and all we have to show is that for a given $m$ the following inequality holds
\begin{align}
&\notag \Bigg|\frac{A^{m+1}(-(m+1)-p+q)}{(1-A)(1-A^{m+2})}  - \frac{A^{m}(-m+-p+q)}{(1-A)(1-A^{m+1})}  + \dfrac{ (m+1)(A^{m+2} -A^{m+1})}{1-A^{m+2}} - \dfrac{ m(A^{m+1} -A^{m})}{1-A^{m+1}}\\  &\frac{A^{m+2}((m+1)+p-q-1)}{(1-A)(1-A^{m+2})} - \frac{A^{m+1}(m+p-q-1)}{(1-A)(1-A^{m+1})} \Bigg| > \notag 
	\Bigg|\textcolor{black}{B}\Big(\frac{1}{1-A^{m+2}} - \frac{1}{1-A^{m+1}}\Big) + \\& \textcolor{black}{C}\Big(\frac{1}{(1-A)(1-A^{m+2})} - \frac{1}{(1-A)(1-A^{m+1})}\Big) +  
	\textcolor{black}{D}\Big(\dfrac{A^{m+2} -A^{m+1}}{1-A^{m+2}}-\dfrac{A^{m+1} -A^m}{1-A^{m+1}}\Big)\Bigg|.
	\label{eq:test}
\end{align} 
From \eqref{eq:test}, it follows that the summation on the right side vanishes as $m$ increases since $B$, $C$, and $D$ are constants, and the limit as the corresponding multiplications tends to infinity is $0$. The condition is $\kappa$ be constant and $\kappa<\infty$. We can use the left side of \eqref{eq:test} to show that $$\lim_{m\rightarrow\infty} \dfrac{E}{F} = 0,$$ 
where $$E=\frac{A^{m+1}(-(m+1)-p+q) + A^{m+2}((m+1)+p-q-1)}{(1-A)(1-A^{m+2})} + \dfrac{ (m+1)(A^{m+2} -A^{m+1})}{1-A^{m+2}} $$ and $$F= \frac{A^{m}(-m-p+q) + A^{m+1}(m+p-q-1)}{(1-A)(1-A^{m+1})} + \dfrac{ m(A^{m}-A^{m+1} )}{1-A^{m+1}} .$$
 Therefore, the left side of \eqref{eq:test} will be greater than the right side for an $m$ large enough. 

Then, for some $m$ large enough, we will find $J_{\pi_{m}}<J_{\pi_{m+1}}$.
\end{proof}
As in \cite{1103637},
using Lemmas \ref{lemma:m>4} and \ref{lemma:m} and Lippman's results (since the sufficient conditions of \cite[Corollary 1]{10.2307/2629409} are satisfied) we next show that the limit as $\beta\rightarrow1^-$ of a sequence of threshold policies is a threshold policy again.

 \begin{theorem}	
The optimal policy for the average cost problem is of threshold type with a finite threshold.
\label{theorem:average}
\end{theorem}

\begin{proof}
From Lemma \ref{lemma:m} we can infer that for some $x\in X$, we have $J^\beta_{\pi_m}(x)  < J^\beta_{\pi_{m+1}}(x) $ for all $\beta$ sufficiently close to 1.
 Theorem \ref{theorem:infinite}  
asserts that a threshold policy 
is optimal for the discounted cost problem. Therefore, some policy in the set of threshold policies $\{\pi_1, \pi_2, \ldots, \pi_m\}$ is optimal for each discount factor $\beta$ sufficiently close to $1$. Since now there exists a threshold policy $\pi_k$, which is optimal for each discount factor $\beta_n$, with $\beta_n \rightarrow 1$, by \cite[Theorem 6]{10.2307/2629409} the average cost problem has an optimal policy which is a member of the same set $\{\pi_1, \pi_2, \ldots, \pi_m\}$. Therefore, the optimal policy is of threshold type.
\end{proof}

Theorem \ref{theorem:average} provides an algorithm to determine an optimal threshold: starting with $\pi_0$, we use 	\eqref{eq:finalJm2} to determine a $m$ such that
$$J_{\pi_0}\geq J_{\pi_1}\geq J_{\pi_2}\geq\cdots\geq J_{\pi_m} \text{ and }  J_{\pi_m} < J_{\pi_{m+1}}.$$
Then $\pi_m$ is the optimal policy.

\section{Numerical results and Discrete event simulation (DES)}\label{sec:results}

In this section, we present the numerical computation of $J_{\pi_m}$ in equation \eqref{eq:finalJm2} and analyze some discrete event simulations. We analyze  values for $p$, $q$, and $\kappa$ considering three different scenarios: i) trucks and platoons arrive with the same probability, ii) platoons arrive with twice as high probability compared to the trucks, and iii) platoons arrive with  approximately $45\%$ more probability than trucks. These scenarios give a general idea of the system's behaviour. For instance,  by increasing $\kappa$ or $q$ (resp. decreasing $p$), we are essentially increasing the threshold $m$. %Due to space constraints, we do not present numerical results for $p>q$.

In Figures \ref{fig:thr0.5}, \ref{fig:thr0.4}, and \ref{fig:thr0.45}, we present the computation of $J_{\pi_m}$, where we vary $m$. The line in blue graphs the numerical computation of $J_{\pi_m}$. The vertical line in orange highlights the optimal threshold. 

In Figures \ref{fig:sim0.5}, \ref{fig:sim0.4}, and \ref{fig:sim0.45}, we show the results of $30$ discrete event simulations that ran for $1,000,000$ steps each. The line in orange indicates the cost of the optimal policy $\pi_m$. 
The blue dots show the outcome of each simulation, and the blue line illustrates the mean of all simulations. The interval in light blue represents the 99\% confidence interval.

Figure \ref{fig:all0.5} presents the results for $p=q=0.5$ and $\kappa=10$.  In Figure \ref{fig:thr0.5}, the optimal threshold is $m=1$. Using $m = 1$, we ran the respective DES, as shown in Figure  \ref{fig:sim0.5}. The mean of the discrete simulations overlaps the optimal value of $J_{\pi_m}$. It follows from \ref{fig:thr0.5} that $J_{\pi_m}$ increases as $m$ increases, which is consistent with Lemma \ref{lemma:asymptotic}.

\begin{figure}[h]
\centering
\begin{subfigure}[t]{.5\textwidth}
  \centering
  \includegraphics[width=1\linewidth,trim={0cm 0cm 0 1.1cm},clip]{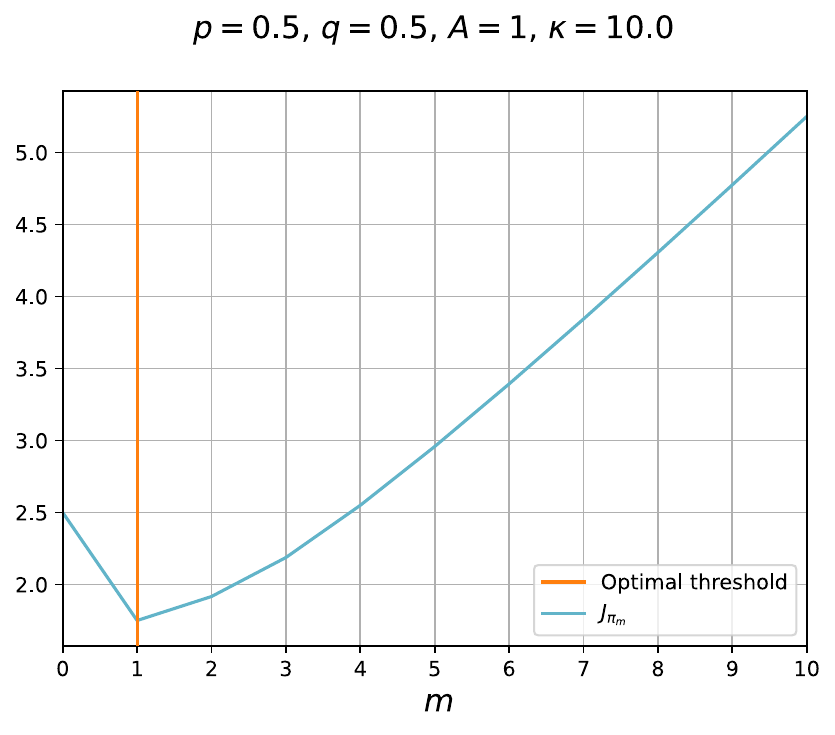}
  \caption{Theoretical computation  $J_{\pi_m}$.}
  \label{fig:thr0.5}
\end{subfigure}\begin{subfigure}[t]{.5\textwidth}
  \centering
  \includegraphics[width=1\linewidth,trim={0cm 0cm 0 .75cm},clip]{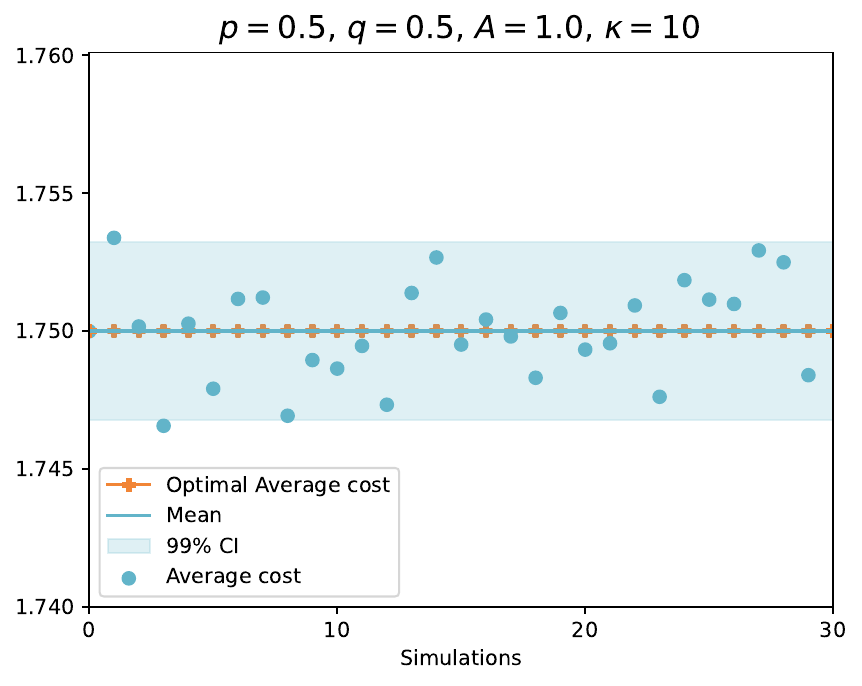}
  \caption{Discrete event simulation.}
  \label{fig:sim0.5}
\end{subfigure}
\caption{For $p=0.5$, $q=0.5$, and $\kappa=10$.}
\label{fig:all0.5}
\end{figure}
Figure \ref{fig:all0.4} presents the results for $p=0.4,$ $q=0.8$ and $\kappa=5$. In Figure \ref{fig:thr0.4}, the optimal threshold is $m=2$. 
In contrast to Figure \ref{fig:thr0.5}, $J_{\pi_m}$ appears to be ``insensitive'' as $m$ increases above the optimal threshold. Since platoons arrive with a high probability relative to cars, and because it is optimal to dispatch a car with an arriving platoon (according to Lemma~\ref{lemma:dispatch}), we expect the queue of waiting trucks to be very small. So, most of the time, the queue size does not exceed the threshold. 
In Figure \ref{fig:thr0.4}, the mean and optimal average costs again overlap.

\begin{figure}[h]
\centering
\begin{subfigure}[t]{.5\textwidth}
  \centering
  \includegraphics[width=1\linewidth,trim={0cm 0cm 0 1.1cm},clip]{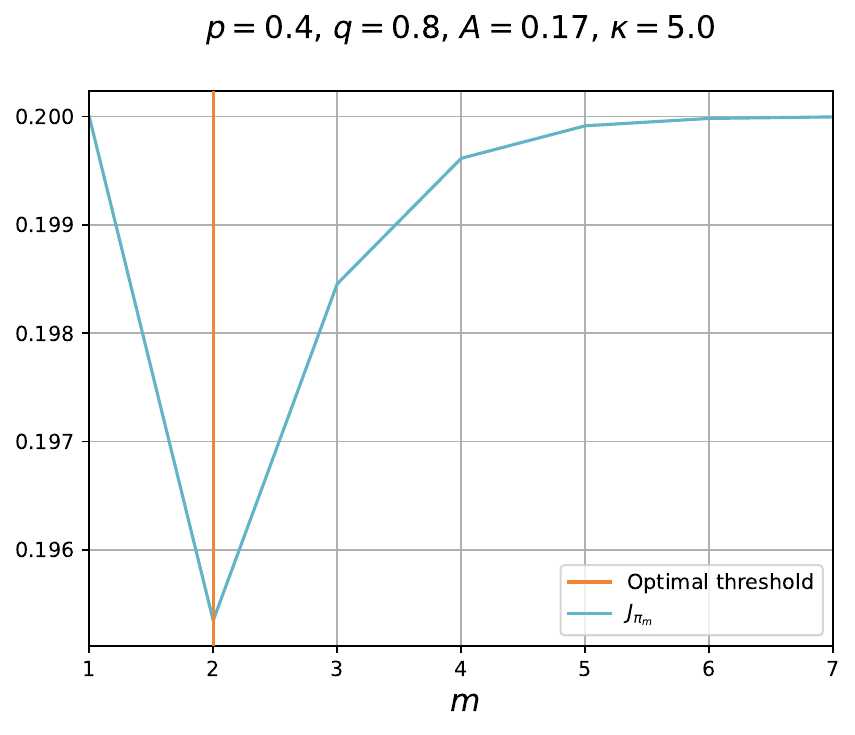}
  \caption{Theoretical computation  $J_{\pi_m}$.}
  \label{fig:thr0.4}
\end{subfigure}%
\begin{subfigure}[t]{.5\textwidth}
  \centering
  \includegraphics[width=1\linewidth,trim={0cm 0cm 0 .75cm},clip]{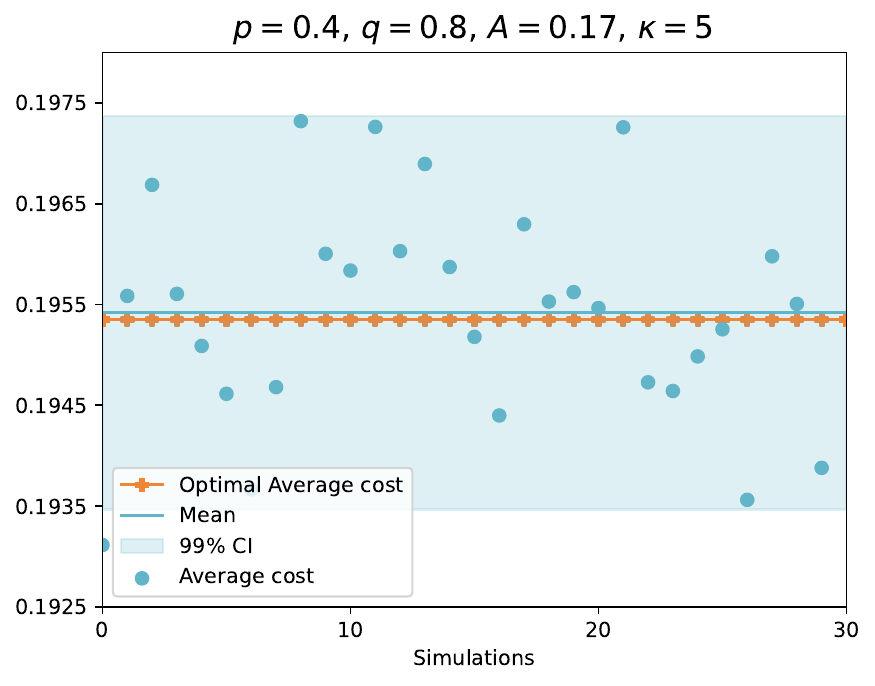}
  \caption{Discrete event simulation.}
  \label{fig:sim0.4}
\end{subfigure}%
\caption{For $p=0.4$, $q=0.8$, and $\kappa=5$.}
 \label{fig:all0.4}
\end{figure}

In Figure \ref{fig:all0.45}, we present the results for $p=0.45,$ $q=0.65$ and $\kappa=20$. These results are also consistent with Lemma \ref{lemma:asymptotic} when $p<q$. The threshold is $m=4$ from Figure  \ref{fig:thr0.45}.  The mean of DES costs and the optimal average cost again overlap each other in Figure \ref{fig:sim0.45}. 

\begin{figure}[h]
\begin{subfigure}{.5\textwidth}
  \centering
  \includegraphics[width=1\linewidth,trim={0cm 0cm 0 1.1cm},clip]{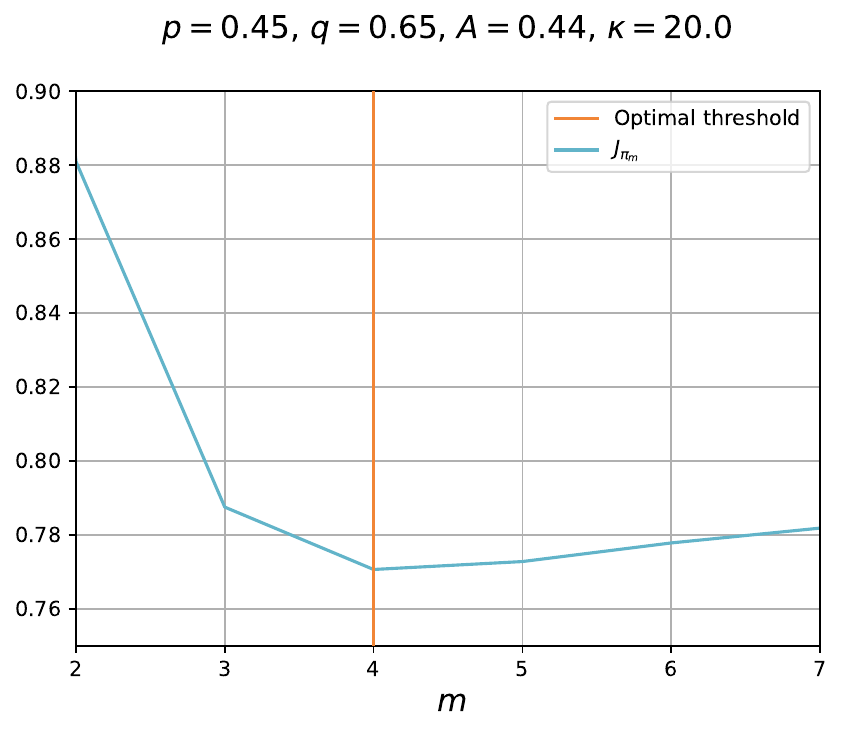}
  \caption{Theoretical computation  $J_{\pi_m}$.}
  \label{fig:thr0.45}
\end{subfigure}%
\begin{subfigure}{.5\textwidth}
  \centering
    \includegraphics[width=1\linewidth,trim={0cm 0cm 0 .75cm},clip]{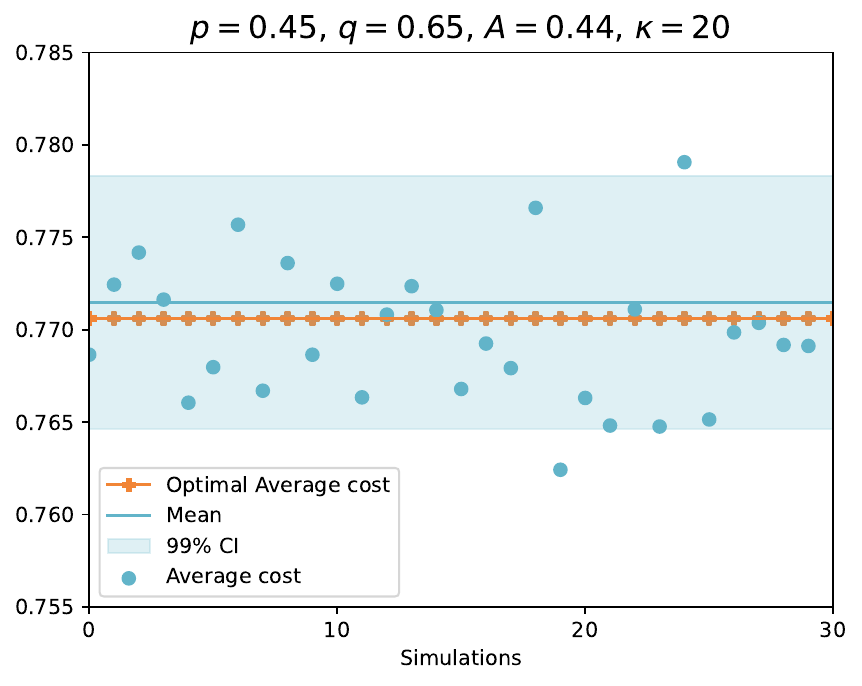}
  \caption{Discrete event simulation.}
  \label{fig:sim0.45}
\end{subfigure}%
\caption{For $p=0.45$, $q=0.65$, and $\kappa=20$.}
 \label{fig:all0.45}

\end{figure}

\section{Conclusion and Future Research}\label{sec:conclusion}
\ifbulletlist
{\color{red} 
\begin{enumerate}
    \item We summarize the paper with our contributions.
    \item Future work: main focus on platoons of different sizes/ different saving costs.
\end{enumerate}
}
\fi %bulletlist

In this paper, we studied the problem of dispatching trucks to  platoons arriving at a highway station. We modelled the problem as a Markov Decision Problem (with a one-dimensional Markov Chain). The dispatching action aims   to minimize a cost function at the station. This function can be $\beta$-discounted (with finite or infinite horizon) as well as  long-run average.  We first proved that for all cost types, the optimal policy   is a threshold-type policy. This threshold is finite for the finite horizon discounted cost, average cost criteria, and for the infinite horizon discounted cost when $\beta$ is sufficiently close to $1$.  We then presented some numerical results regarding the optimal policy. 

The results can be generalized in several ways in future work. For instance, we could  
consider investigating a batch dispatching model where we empty the queue upon a platoon's arrival. More generally, we could send $r$ trucks from the station or empty it when $r \geq x_n$. In both models, the cost of platooning could be introduced since we assume no cost for platooning in the current model. We suspect that the optimal policy for the batch dispatching model is also of threshold type. Perhaps with a greater threshold than the current model for the same parameters since we can now send multiple trucks from the queue.

Another approach would be to investigate a model where dispatching a truck to a platoon %\ek{change: "provide" to "may have"} 
may have random savings or costs. Finally, we could consider a model where we can dispatch $r_n$ trucks from the station, where $r_n$ will depend on the available space on an arriving platoon.

%\bibliographystyle{unsrtnat}
%\bibliography{draft_arxiv.bib}
%\begin{comment}

%\end{comment}

\end{document}